\documentclass[11pt]{amsart}
\usepackage{amsmath,amsxtra,amssymb,latexsym, amscd,amsthm}
\usepackage[colorlinks]{hyperref}
\usepackage{tikz}
\usepackage{graphicx}
\numberwithin{equation}{section}
\newtheorem{theorem}{Theorem}[section]
\newtheorem{lemma}[theorem]{Lemma}
\newtheorem{proposition}[theorem]{Proposition}

\newtheorem{definition}[theorem]{Definition}

\newtheorem{remark}[theorem]{Remark}

\begin{document}
\title[Decorrelation estimates]{Decorrelation estimates for a 1D 
tight \\ binding model in the localized regime}         
\author{Trinh Tuan Phong}       
\date{\today}         
\begin{abstract}
\noindent In this article, we prove decorrelation estimates for the eigenvalues of a 1D discrete tight binding model near two distinct energies in the localized regime. Consequently, for any integer $n\geq 2$, the asymptotic independence for local level statistics near $n$ distinct energies is obtained.
\end{abstract}
\maketitle

\section {Introduction}\label{S:intro}
The present paper deals with the following lattice Hamiltonian with off-diagonal disorder in dimension 1: for $u=\{u(n)\}_{n\in \mathbb Z}\in l^2(\mathbb Z),$ set
\begin{equation}\label{eq:1.1}
(H_{\omega}u)(n)=\omega_n (u(n)-u(n+1)) -\omega_{n-1}(u(n-1)-u(n)).
\end{equation}
The model (\ref{eq:1.1}) appears in the description of waves (light, acoustic waves, etc) which propagate through a disordered, discrete medium (c.f. \cite{AFAK94}). We can see $\{\omega_n\}_{n\in \mathbb Z}$ in this model as weights of bonds of the lattice $\mathbb Z.$\\ 
Throughout this article, we assume that $\omega:=\{\omega_n\}_{n\in \mathbb Z}$ are non-negative i.i.d. random variables (r.v.'s for short) with a bounded, compactly supported density $\rho$.\\
In addition, from Section \ref{S:intro} to Section \ref{S:decorrelation}, we assume more that  $\omega_n \in [\alpha_0, \beta_0]$ for all $n \in \mathbb Z$ where $\beta_0> \alpha_0>0.$ 
In Section \ref{S:remark}, we will comment on relaxing the hypothesis of the lower bound of r.v.'s $\omega$.\\ 
It is known that (see \cite{DM2011}): 
\begin{itemize}
\item the operator $H_{\omega}$ admits an almost sure spectrum $\Sigma:=[0,4\beta_0].$
\item  $H_{\omega}$ has an integrated density of states defined as follows:
\vskip 0.5 pt
\noindent $\omega-$a.s., the following limit exists and is $\omega$ independent: 
 \begin{equation}\label{eq:1.2}
N(E):=\lim_{|\Lambda|\rightarrow +\infty} \dfrac{\# \{\text {e.v. of $H_{\omega}$ less than E} \}}{|\Lambda|} \text{ for a.e. $E$}.
\end{equation}
As a direct consequence of the Wegner estimate (see Theorem \ref{T:Wegner} in Section ~\ref{S:Pre}), $N(E)$ is defined everywhere in $\mathbb R$ and absolutely continuous w.r.t. Lebesgue measure with a bounded derivative $\nu(E)$ called the density of states of $H_{\omega}.$  
\end{itemize}
In the present paper, we follow a usual way to study various statistics related to random operators. We restrict the operator $H_{\omega}$ on some interval $\Lambda \subset \mathbb Z$ of finite length with some boundary condition and obtain a finite-volume operator which is denoted by $H_{\omega}(\Lambda).$
Then, we study diverse statistics for this operator in the limit when $|\Lambda|$ goes to infinity.\\ 
Throughout this paper, the boundary condition to define $H_{\omega}(\Lambda)$ is always the periodic boundary condition.   
For example, if $\Lambda=[1,N],$ the operator $H_{\omega}(\Lambda)$ is a symmetric $N\times N$ matrix of the following form:  
\begin{equation*}
\left (
\begin{matrix}
\omega_N+\omega_1 &  -\omega_1 &0 &\dots   & 0 &-\omega_N\\
-\omega_1& \omega_1+\omega_2 & -\omega_2 &\dots &0 &0 \\
\vdots &\vdots &\vdots &\dots &\vdots &\vdots\\  
0 & 0 &0 &\dots & \omega_{N-2}+\omega_{N-1} &-\omega_{N-1}\\
-\omega_N &0 &0& \dots & -\omega_{N-1} &\omega_{N-1}+\omega_N
\end{matrix}
\right )   
\end{equation*}
For $L\in \mathbb N,$ let $\Lambda=\Lambda_{L}:=[-L,L]$ be a large interval in $\mathbb Z$ and $|\Lambda|:=(2L+1)$ be its cardinality.\\
We will denote the eigenvalues of $H_{\omega}(\Lambda)$ ordered increasingly and repeated according to multiplicity by $ E_1(\omega,\Lambda)\leqslant E_2(\omega,\Lambda)\leqslant \cdots \leqslant E_{|\Lambda|}(\omega,\Lambda).$\\
Let $I$ be the localized regime (the region of localization) in $\Sigma$ where the finite-volume fractional-moment criteria for localization are satisfied for the finite-volume operators $H_{\omega}(\Lambda)$ when $|\Lambda|$ is large enough (see Section 2 and \cite{ASFH2001} for more details).
In this region, the spectrum of $H_{\omega}$ is pure point and the corresponding eigenfunctions decay exponentially at infinity.\\
Pick $E$ a positive energy in $I$ with $\nu(E)>0,$ and define the local level statistics near $E$ as follows
\begin{equation}\label{eq:1.3}
\Xi(\xi,E,\omega,\Lambda)=\sum_{n=1}^{|\Lambda|}\delta_{\xi_n}(E,\omega,\Lambda)(\xi)
\end{equation} 
where 
\begin{equation}\label{eq:1.4}
\xi_{n}(E,\omega,\Lambda)=|\Lambda|\nu(E)(E_n(\omega,\Lambda)-E).
\end{equation}
For the model (\ref{eq:1.1}), it is known that the weak limit of the above point process is a Poisson point process:
\begin{theorem}\label{T:oldpoisson}[\cite{DM2011}]
Assume that $E$ is a positive energy in $I$ with $\nu(E)>0.$\\
Then, when $|\Lambda| \rightarrow +\infty,$ the point process $\Xi(\xi,E,\omega,\Lambda)$ converges weakly to a Poisson point process with the intensity $1$ i.e., for $(U_j)_{1\leqslant j\leqslant J},\; U_j \subset \mathbb R$ bounded measurable and $U_{j'}\cap U_j =\emptyset$ if $j\ne j'$ and $(k_j)_{1\leqslant j\leqslant J} \in \mathbb N^{J},$ we have
\begin{equation*}
\lim_{|\Lambda| \rightarrow +\infty} \left | \mathbb P \left (
\begin{Bmatrix}
\# \{j; \xi_j(E,\omega,\Lambda) \in U_1 \}  &= k_1\\ 
\vdots & \vdots\\ 
 \# \{j; \xi_j(E,\omega,\Lambda) \in U_j \}  &= k_J\\
\end{Bmatrix}
\right ) - \prod_{j=1}^J \dfrac{|U_j|^{k_j}}{k_j!} e^{-|U_j|}  \right | =0.
\end{equation*}   
\end{theorem}
\noindent Recently, for the 1D discrete Anderson model, Klopp \cite{FK2011} showed more that if we pick two fixed, distinct energies $E$ and $E'$ in the localized regime, their two corresponding point processes $\Xi(\xi,E,\omega,\Lambda)$ and $\Xi(\xi,E', \omega,\Lambda)$ converge weakly, respectively to two independent Poisson point processes. In other words, the limits of $\Xi(\xi,E,\omega,\Lambda)$ and $\Xi(\xi,E',\omega,\Lambda)$ are stochastically independent.\\
\noindent It is known that the above statement holds true if one can prove a so-called decorrelation estimate.\\
That is exactly what we want to carry out here for the 1D discrete lattice Hamiltonian with off-diagonal disorder (\ref{eq:1.1}). Our decorrelation estimate is the following:
\begin{theorem}\label{T:decorrelation}
Let $E, E'$ be two positive, distinct energies in the localized regime. 
 Pick $\beta \in \left(1/2,1 \right)$ and $\alpha \in \left(0,1\right).$ Then, 
 for any $c>0,$ there exists $C>0$ such that, for $L$ large enough and $cL^{\alpha}\leqslant l\leqslant L^{\alpha}/c,$ one has 
\begin{equation*}
\mathbb P \Big(
\begin{Bmatrix}
\sigma (H_{\omega}(\Lambda_l))\cap (E+L^{-1}(-1,1)) \ne \emptyset\\ 
\sigma (H_{\omega}(\Lambda_l))\cap (E'+L^{-1}(-1,1)) \ne \emptyset\\ 
\end{Bmatrix}
\Big ) \leqslant C (l/L)^{2}e^{(\log L)^{\beta}}.
\end{equation*}
\end{theorem} 
\noindent This decorrelation estimate means that, up to a sub-polynomial error, the probability of obtaining simultaneously two eigenvalues near two distinct energies is bounded by the product of the estimates given by Wegner estimate for each of these two energies. Roughly speaking, two eigenvalues of our model near two distinct energies behave like two independent random variables.
\\
Thanks to Theorem \ref{T:decorrelation}, we can proceed as in Section 3 of \cite{FK2011} to obtain the asymptotic independence of the weak limits of $\Xi(\xi,E,\omega,\Lambda)$ and $\Xi(\xi,E',\omega,\Lambda)$ with $E, E' >0$ for the model (\ref{eq:1.1}): 
\begin{theorem}\label{T:Poisson}
Pick two positive, distinct energies $E$ and $E'$ in the localized regime such that $\nu(E)>0$ and  $ \nu(E')>0$.\\
When $|\Lambda| \rightarrow +\infty,$ the point processes $\Xi(\xi,E,\omega,\Lambda),$ and $\Xi(\xi,E',\omega,\Lambda)$ converge weakly respectively to two independent Poisson processes on $\mathbb R$ with intensity the Lebesgue measure.\\ 
That is, for $(U_j)_{1\leqslant j\leqslant J},\; U_j \subset \mathbb R$ bounded measurable and $U_{j'}\cap U_j =\emptyset$ if $j\ne j'$ and $(k_j)_{1\leqslant j\leqslant J} \in \mathbb N^{J}$ and $(U'_j)_{1\leqslant j\leqslant J'},\; U'_j \subset \mathbb R$ bounded measurable and $U'_{j'}\cap U'_j =\emptyset$ if $j\ne j'$ and $(k'_j)_{1\leqslant j\leqslant J'} \in \mathbb N^{J'},$ we have
\begin{align}\label{eq:poisson}
\mathbb P \left (
\begin{Bmatrix}
\# \{j; \xi_j(E,\omega,\Lambda) \in U_1 \}  &= k_1\\ 
\vdots & \vdots\\ 
 \# \{j; \xi_j(E,\omega,\Lambda) \in U_j \}  &= k_J\\
\# \{j; \xi_j(E',\omega,\Lambda) \in U'_1 \}  &= k'_1\\ 
\vdots & \vdots\\ 
\# \{j; \xi_j(E',\omega,\Lambda) \in U'_J \}  &= k'_J
\end{Bmatrix}
\right ) 
\rightarrow \prod_{j=1}^J \dfrac{|U_j|^{k_j}}{k_j!} e^{-|U_j|} \prod_{i=1}^{J'} \dfrac{|U'_i|^{k'_i}}{k'_i!}e^{-|U'_i|}  
\end{align}
as $|\Lambda| \rightarrow +\infty$.   
\end{theorem}
\noindent Moreover, in Section ~\ref{S:morethan2}, we will generalize Theorem \ref{T:Poisson} by considering not only two but any fixed number of distinct energies.\\  
\noindent To prove Theorem ~\ref{T:decorrelation} for the model (\ref{eq:1.1}), we follow the strategy introduced in ~\cite{FK2011}.
The key point of the proof of decorrelation estimates for the 1D discrete Anderson model in ~\cite{FK2011} is to derive that 
the gradient with respect to $\omega$ of two eigenvalues $E(\omega)$ and $E'(\omega)$ near two distinct energies $E$ and $E'$ are not co-linear with a good probability. In the discrete Anderson case, this statement will hold true if the gradients of $E(\omega)$ and $E'(\omega)$ are distinct. 
In deed, the gradients of $E(\omega)$ and $E'(\omega)$ have non-negative components and their $l^1-$norm are always equal to $1$. So, if they are co-linear, they should be the same.
\vskip 1pt
\noindent Unfortunately, the $l^1-$norm of the gradient of an eigenvalue of our finite-volume operator $H_{\omega}(\Lambda)$ is not a constant w.r.t. $\omega$ anymore (it is even not bounded from below by a positive constant uniformly w.r.t. $\omega$). 
Moreover, to prove the above key point for the discrete Anderson model, ~\cite{FK2011} exploits the diagonal structure of the potential which can not be used for the present case. So, a different approach in the proof is needed to carry out Theorem \ref{T:decorrelation}. This approach is contained in Lemma \ref{L:mainlemma}.
\vskip 1pt
\noindent In addition, 
Theorem 1.12 in ~\cite{FGFK} implies directly the following result for the model (\ref{eq:1.1}):
\begin{theorem}\label{T:fur} [Theorem 1.12, ~\cite{FGFK}]  
Pick $0<E_0 \in I$ such that the density of states $\nu$ is continuous and positive at $E_0.$\\
Consider two sequences of positive energies, say $(E_{\Lambda})_{\Lambda},$ $(E'_{\Lambda})_{\Lambda}$ such that 
\begin{enumerate}
\item $E_{\Lambda}\xrightarrow[\Lambda\rightarrow \mathbb Z^d]{} E_0$ and $E'_{\Lambda}\xrightarrow[\Lambda\rightarrow \mathbb Z^d]{} E_0,$
\item $|\Lambda| |N(E_{\Lambda})-N(E'_{\Lambda})|\xrightarrow[\Lambda\rightarrow \mathbb Z^d]{}+\infty.$
\end{enumerate}
Then, the point processes $\Xi(\xi,E_{\Lambda},\omega,\Lambda)$ and $\Xi(\xi,E'_{\Lambda},\omega,\Lambda)$ converges weakly respectively to two independent Poisson point processes in $\mathbb R$ with intensity the Lebesgue measure. 
\end{theorem} 
\noindent In Theorem \ref{T:fur}, instead of fixing two distinct energies $E$ and $E'$, one considers two sequences of positive energies $\{E_{\Lambda}\},$ $\{E'_{\Lambda}\}$ which tend to each other as $|\Lambda|\rightarrow \infty$. In addition, one assumes that, roughly speaking, there are sufficiently many eigenvalues of $H_{\omega}(\Lambda)$ between $E_{\Lambda}$ and $E'_{\Lambda}$ as $\Lambda$ large. Then, the asymptotic independence of two point processes associated to $E_{\Lambda}$ and $E'_{\Lambda}$ is obtained. (See Section \ref{S:Pre} for the reason we only consider positive energies in the present paper).\\   
Besides, it is known that the existence of an integrated density of states defined as in (~\ref{eq:1.2}) implies that the average distance (mean spacing) between eigenlevels is of order $|\Lambda|^{-1}.$\\  
Thus, according to Theorem ~\ref{T:fur}, in the localized regime, eigenvalues separated by a distance that is asymptotically infinite with respect to the mean spacing between eigenlevels behave like independent random variables. In other words, there are no interactions between distinct eigenvalues, except at a very short distance.\\
\text{\bf Notation:} In the present paper, we use Dirac's notations: If $\varphi$ is a vector in a Hilbert space $\mathcal H,$ we denote by $|\varphi\rangle\langle \varphi| =\langle \varphi, \cdot\rangle_{\mathcal H} \varphi $ the projection operator on $\varphi.$
Besides, throughout the present paper, the symbol $\|\cdot\|$ stands for the $l^2-$norm $\|\cdot\|_{2}$ in some finite dimensional Hilbert space.
\section{Preliminaries}\label{S:Pre}
We recall here a Wegner-type estimate and a Minami-type estimate for the model (\ref{eq:1.1}) which are essentially important for the proofs of Poisson statistics as well as those of decorrelation estimates.
\begin{theorem} \label{T:Wegner}[Wegner estimate, Theorem 2.1, \cite{DM2011}] 
$$ \mathbb P( \text{dist} (E, \sigma(H_{\omega}(\Lambda))) \leqslant \epsilon) \leq \dfrac{2d \|s\rho(s)\|_{\infty}}{E-\epsilon} \epsilon |\Lambda|
 $$
for all intervals $\Lambda\subset \mathbb Z$ and $0<\epsilon <E$.
\end{theorem} 
\noindent Roughly speaking, the above Wegner estimate means that the probability that we can find an eigenvalue of the finite-volume operator  $H_{\omega}(\Lambda)$ in a interval (suppose that the length of this interval is small) is bounded by the length of this interval times $|\Lambda|.$\\  
Following is the Minami estimate:
\begin{theorem}\label{T:Minami}[Minami estimate, Theorem 3.1, \cite{DM2011}] 
There exists $C>0$ such that, for all intervals $J=[a,b] \subset (0, +\infty),$ and $\Lambda \subset \mathbb Z,$ we have
$$ \mathbb P (\# \{ \sigma (H_{\omega}(\Lambda))\cap J  \}   \geqslant 2 ) \leqslant \beta_0 \|\rho\|_{\infty}\|s\rho(s)\|_{\infty} (|J||\Lambda|)^2/2a^2. $$
\end{theorem}
\noindent We should remark that although the almost sure spectrum $\Sigma$ of the operator (\ref{eq:1.1}) contains the point $0$ (the bottom of $\Sigma$), the above Wegner and Minami estimates do not work at $0$. That is why, in the present paper,  we will only consider energies who are strictly positive in $\Sigma$. Besides, in Theorem \ref{T:Minami}, the right hand side of the Minami estimate depends on the infimum of the interval $J$ which is different from the Minami estimate stated for the discrete Anderson model. However, this difference complicates nothing when we only consider energies close to a fixed, positive energy in the present paper.    
\vskip 1pt 
\noindent Finally, we would like to remind readers of  the precise definition of the localized regime associated to finite-volume operators.
\begin{proposition} \label{P:localised}(c.f. \cite{ASFH2001}, \cite{FK2011})
Let $I$ be the region of $\Sigma$ where the finite volume fractional moment criteria of \cite{ASFH2001} for $H_{\omega}(\Lambda)$ are verified for $\Lambda$ sufficiently large.\\
Then, there exists $\nu>0$ such that, for any $p>0,$ there exists $q>0$ and $L_0>0$ such that, for $L\geqslant L_0,$ with probability larger than $1-L^{-p},$ if 
\begin{enumerate}
\item $\varphi_{n,\omega}$ is a normalized eigenvector of $H_{\omega}(\Lambda_L)$ associated to an energy $E_{n,\omega}\in I,$
\item $x_{n,\omega}\in \Lambda_L$ is a maximum of $x\mapsto |\varphi_{n,\omega}(x)|$ in $\Lambda_L,$
\end{enumerate} 
then, for $x\in \Lambda_L,$ one has 
$$ |\varphi_{n,\omega}(x)|\leqslant L^q e^{-\nu|x-x_{n,\omega}|}. $$
The point $x_{n,\omega}$ is called a localization center for $\varphi_{n,\omega}$ or $E_{n,\omega}.$
\end{proposition} 
\noindent Remark that an eigenvalue of $H_{\omega}(\Lambda)$ can have more than one localization center. Nevertheless, it is not hard to check that, under assumptions of Proposition \ref{P:localised}, all localization centers associated to one eigenvalue are contained in a disk of radius $C \log L$ with some constant $C$. 
\vskip 0.3 pt
\noindent Hence, to define the localization center uniquely for an eigenvalue, we can order these centers lexicographically and choose the localization center associated to an eigenvalue to be the largest one.   
\section{Proof of Theorem \ref{T:decorrelation}} \label{S:decorrelation}
Pick two distinct, positive energies $E, E'$ in $I$ (the localized regime). Let $J_L=E+L^{-1}[-1,1]$ and $J'_{L}=E'+L^{-1}[-1,1]$ with $L$ large. 
We would like to begin this section by proving some elementary properties of eigenvalues of $H_\omega(\Lambda)$ with an arbitrary interval $\Lambda \in \mathbb Z.$
\begin{lemma} \label {L:perturbation}
Suppose that $\omega \mapsto E(\omega)$ is the only eigenvalue of $H_\omega(\Lambda)$ in $J_L.$  Then
\begin{enumerate}
\item $E(\omega)$ is simple and $\omega \mapsto E(\omega)$ is real analytic. Moreover, let $\omega \mapsto \varphi(\omega)$ denote the real-valued, normalized eigenvector associated to $E(\omega)$, it is also real analytic in $\omega$. 
\item $ \|\nabla _{\omega}E(\omega)\|_{1} =2\sum_{\gamma \in \Lambda} \|\Pi_{\gamma} \varphi\|^2$ where  
$ \Pi_{\gamma}= \dfrac{1}{2}|\delta_{\gamma}-\delta_{\gamma+1}\rangle \langle \delta_{\gamma}-\delta_{\gamma+1}|$ is a projection in $l^2(\Lambda).$ Besides, we have $E(\omega) \in [0, 4\beta_0]$. 
\item Hess$_{\omega}E(\omega)=(h_{\gamma,\beta})_{\gamma,\beta}$ where
\begin{itemize}
\item  $h_{\gamma,\beta}:=-4 
\big \langle \big(H_{\omega}(\Lambda) -E(\omega)\big)^{-1} \psi_{\gamma}, \psi_{\beta} \big \rangle .$
\item $\psi_{\gamma}:= \langle \Pi_{\gamma}\varphi ,\varphi \rangle \varphi-\Pi_{\gamma}\varphi=-\Pi_{\langle \varphi\rangle^{\perp}} (\Pi_{\gamma} \varphi)$ 
where $\Pi_{\langle \varphi\rangle^{\perp}}$ is the orthogonal projection on $\langle \varphi\rangle^{\perp}.$  
\end{itemize}
\end{enumerate}
\end{lemma}
\begin{proof}[Proof of Lemma \ref{L:perturbation}]
(1) is true from the standard perturbation theory (c.f. \cite{TK1995}).\\
Now we will prove (2). Starting from the eigenequation
\begin{equation}\label{eq:3.1}
H_{\omega}(\Lambda) \varphi =E(\omega) \varphi,
\end{equation}
we have, for all $\gamma \in \Lambda,$
\begin{alignat*}{2}
\partial_{\omega_{\gamma}} E(\omega) &= \langle \partial_{\omega_{\gamma}}(H_{\omega}(\Lambda)\varphi), \varphi \rangle +
 \langle H_{\omega}(\Lambda)\varphi, \partial_{\omega_{\gamma}}\varphi \rangle\\ 
&=  \langle \partial_{\omega_{\gamma}}(H_{\omega}(\Lambda)) \varphi, \varphi \rangle + \langle H_{\omega}(\Lambda) \partial_{\omega_{\gamma}}\varphi, \varphi \rangle + \langle H_{\omega}(\Lambda)\varphi, \partial_{\omega_{\gamma}}\varphi \rangle\\
&= \langle \partial_{\omega_{\gamma}}(H_{\omega}(\Lambda)) \varphi, \varphi \rangle +\langle \partial_{\omega_{\gamma}}\varphi, 
H_{\omega}(\Lambda)\varphi\rangle + \langle H_{\omega}(\Lambda)\varphi, \partial_{\omega_{\gamma}}\varphi \rangle\\
&= \langle \partial_{\omega_{\gamma}}(H_{\omega}(\Lambda)) \varphi, \varphi \rangle+ E(\omega) \big(\langle \partial_{\omega_{\gamma}} \varphi, \varphi \rangle +\langle \varphi,\partial_{\omega_{\gamma}} \varphi\rangle \big) 
\end{alignat*}
where the last two equalities come from the symmetry of $H_{\omega}(\Lambda)$ and (\ref{eq:3.1}).\\
Noting that   
$$ \langle \partial_{\omega_{\gamma}} \varphi, \varphi \rangle +\langle \varphi,\partial_{\omega_{\gamma}} \varphi\rangle =2 \partial_{\omega_{\gamma}} \|\varphi\|^2 =0 .$$
Hence,
\begin{equation}\label{eq:3.2}
 \partial_{\omega_{\gamma}} E(\omega)= \langle \partial_{\omega_{\gamma}}(H_{\omega}(\Lambda)) \varphi, \varphi \rangle =
2 \langle \Pi_{\gamma}\varphi , \varphi \rangle.
\end{equation}
On the other hand, it is easy to check that $\Pi_{\gamma} =\Pi^{*}_{\gamma}=\Pi_{\gamma}^2.$ Hence,   
$\Pi_{\gamma}$ is an orthogonal projection and $\partial_{\omega_{\gamma}} E(\omega)=2\|\Pi_{\gamma} \varphi\|^2.$\\
Thanks to (\ref{eq:3.1}) and (\ref{eq:3.2}), we have the following important equality: 
\begin{equation} \label{eq:**}
\sum_{\gamma \in \Lambda} \omega_{\gamma} \partial_{\omega_{\gamma}} E(\omega) = 2 \sum_{\gamma \in \Lambda} \omega_{\gamma} \langle \Pi_{\gamma}\varphi, \varphi\rangle=E(\omega)
\end{equation}
which characterize the form of our operator.\\
From (\ref{eq:**}) and $\|\varphi\|=1,$ we infer that 
\begin{equation} \label{eq:bound}
0 \leq E(\omega) = \sum_{\gamma \in \Lambda} \omega_{\gamma}(\varphi(\gamma)-\varphi(\gamma+1))^2 \leq 4 \beta_0.
\end{equation}
Finally, we give a proof for (3). By differentiating both sides of (\ref{eq:3.2}) w.r.t. $\omega_{\gamma},$ we have
\begin{align}\label{eq:3.3}
\partial_{\omega_{\gamma}}^2 E(\omega)& =2 \langle \partial_{\omega_{\gamma}} (\Pi_{\gamma}\varphi),\varphi \rangle + 2\langle \Pi_{\gamma}\varphi,  \partial_{\omega_{\gamma}} \varphi \rangle \\
&= 2\langle \Pi_{\gamma} \partial_{\omega_{\gamma}}\varphi,\varphi\rangle + 2\langle \Pi_{\gamma}\varphi,  \partial_{\omega_{\gamma}} \varphi \rangle =4 
\langle \Pi_{\gamma}\varphi,  \partial_{\omega_{\gamma}} \varphi \rangle  \notag
\end{align}
Next, we will compute $ \partial_{\omega_{\gamma}} \varphi.$\\
Differentiating both sides of (\ref{eq:3.1}) with respect to $\omega_{\gamma}$ to get 
\begin{alignat*}{2}
\big(\partial_{\omega_{\gamma}}H_{\omega}(\Lambda) \big)\varphi +H_{\omega}(\Lambda) \partial_{\omega_{\gamma}}\varphi
&= \partial_{\omega_{\gamma}} E(\omega) \varphi + E(\omega) \partial_{\omega_{\gamma}} \varphi \\
&= 2\langle \Pi_{\gamma}\varphi ,\varphi \rangle \varphi + E(\omega) \partial_{\omega_{\gamma}}\varphi
\end{alignat*}
Therefore,
$$[H_{\omega}(\Lambda) -E(\omega)] \partial_{\omega_{\gamma}}\varphi= 2 \langle \Pi_{\gamma}\varphi ,\varphi \rangle \varphi -\big(\partial_{\omega_{\gamma}}H_{\omega}(\Lambda) \big)\varphi = 2 \Big(\langle \Pi_{\gamma}\varphi ,\varphi \rangle \varphi-\Pi_{\gamma}\varphi\Big).$$
Observe that $\psi_{\gamma}:= \langle \Pi_{\gamma}\varphi ,\varphi \rangle \varphi-\Pi_{\gamma}\varphi \in \langle \varphi\rangle^{\perp},$ and $[H_{\omega}(\Lambda) -E(\omega)] $ is invertible in the subspace $\langle \varphi\rangle^{\perp}$ of $l^2(\Lambda),$ we get 
\begin{equation}\label{eq:3.4}
\partial_{\omega_{\gamma}} \varphi = 2  \big(H_{\omega}(\Lambda) -E(\omega)\big)^{-1} \big( \langle \Pi_{\gamma}\varphi ,\varphi \rangle \varphi-\Pi_{\gamma}\varphi \big).
\end{equation}
From (\ref{eq:3.3}) and (\ref{eq:3.4}), we obtain 
$$\partial_{\omega_{\gamma}}^2 E(\omega)  = 4 
\big \langle \Pi_{\gamma}\varphi, \big(H_{\omega}(\Lambda) -E(\omega)\big)^{-1} \big( \langle \Pi_{\gamma}\varphi ,\varphi \rangle \varphi-\Pi_{\gamma}\varphi \big)  \big \rangle .$$
Thanks to (\ref{eq:3.4}), we have $ 2\big(H_{\omega}(\Lambda) -E(\omega)\big)^{-1} \big( \langle \Pi_{\gamma}\varphi ,\varphi \rangle \varphi-\Pi_{\gamma}\varphi \big)$ is orthogonal to $\varphi.$ We therefore infer that 
$$ \partial_{\omega_{\gamma}}^2 E(\omega) = -4 
\big \langle \big(H_{\omega}(\Lambda) -E(\omega)\big)^{-1} \psi_{\gamma}, \psi_{\gamma} \big \rangle .$$
Repeating this argument, it is not hard to  prove that  
$$ \partial_{\omega_{\gamma}\omega_{\beta}}^2 E(\omega) = -4 
\big \langle \big(H_{\omega}(\Lambda) -E(\omega)\big)^{-1} \psi_{\gamma}, \psi_{\beta} \big \rangle $$
for all $\gamma, \beta.$ 
So, we have Lemma \ref{L:perturbation} proved.
\end{proof}
\noindent Next, assume that $E(\omega)$ is an eigenvalue of $H_{\omega}(\Lambda)$ with $\Lambda=[-L,L]$. Recall that $\omega_j \in [\alpha_0, \beta_0] \text{ for all } j\in \mathbb Z.$
From Lemma \ref{L:perturbation}, we have $E(\omega) \in [0, 4\beta_0].$
Denote by $u:=u(\omega)$ the normalized eigenvector associated to $E(\omega).$ We would like to prove a "lower bound" for $u$ in the sense that there exists a large subset $J$ in $\Lambda$ such that the components $(u(k))_{k\in J}$ of $u$ can not be too small.  
\begin{lemma} \label{L:lowerbound} 
Pick $\beta \in (1/2,1).$ Then, there exists a point $k_0$ in $\Lambda$ and  a positive constant  $\kappa$ depending only on $\alpha_0, \beta_0$ such that 
$$u^2(k)+u^2(k+1)\geqslant e^{-L^\beta/2}$$ 
for all $|k-k_0|\leqslant \kappa L^\beta$ when L is large enough.
\end{lemma}
\begin{proof}[Proof of Lemma \ref{L:lowerbound}]
We rewrite the eigenequation corresponding to the\\
 eigenvector $u$ and eigenvalue $E(\omega)$ at the point $n$ by means of the transfer matrix
\begin{equation*}
\begin{pmatrix}
u(n+1)\\ 
u(n)\\ 
\end{pmatrix}  
=\begin{pmatrix}
\dfrac{\omega_n +\omega_{n-1}-E(\omega)}{\omega_n}&\dfrac{-\omega_{n-1}}{\omega_n}\\ 
1&0\\ 
\end{pmatrix}
\begin{pmatrix}
u(n)\\ 
u(n-1)\\ 
\end{pmatrix} . 
\end{equation*}
Now, let $T(n,E(\omega))$ and $v(n)$  denote the transfer matrix 
\begin{equation*}
\begin{pmatrix}
\dfrac{\omega_n +\omega_{n-1}-E(\omega)}{\omega_n}&\dfrac{-\omega_{n-1}}{\omega_n}\\ 
1&0\\ 
\end{pmatrix}
\end{equation*}
 and the column vector $(u(n+1),u(n))^{t}$ respectively. \\  
Then, for all $n$ greater than $m$ we have
$$v(n) = T(n,E(\omega))\cdots T(n-m+1,E(\omega))v(m).$$
It is easy to check that the transfer matrices are invertible. Moreover, since $E(\omega) \in [0, 4\beta_0]$ and $0<\alpha_0 \leq \omega_j \leq \beta_0$, they and their inverse matrices 
are uniformly bounded  by a constant $C>1$ depending only on  $\alpha_0$ and $\beta_0$.\\ 
Thus, 
$$\|v(n)\|\leqslant C^{|n-m|} \|v(m)\| = e^{(\log C) |n-m|}\|v(m)\| =e^{\eta |n-m|}\|v(m)\|$$
for all $n, m$ in $\Lambda$ with $\eta=\log C>0.$\\
Assume that $\|v(k_0)\|$ is the maximum of $\|v(n)\|.$ Hence, 
$$\|v(k_0)\|\geqslant \dfrac{1}{\sqrt{L}}$$
from the fact that $\sum_{j\in \Lambda} \|v(j)\|^2 =2.$
Thus, for any $\kappa>0$, the following holds true
$$ \|v(k)\| \geqslant  \dfrac{1}{\sqrt{L}} e^{-\eta|k-k_0|} \geq e^{-2\kappa \eta L^{\beta}} $$
for $|k-k_0|\leqslant  \kappa L^{\beta}$ and $L$ sufficiently large.\\
In other words, we have
$$u^2(k)  +u^2(k+1) \geqslant e^{-4\kappa \eta L^{\beta}}$$
for all $|k-k_0|\leqslant  \kappa L^{\beta}.$
So, by choosing $ \kappa=\dfrac{1}{8\eta},$ we have Lemma \ref{L:lowerbound} proved.
\end{proof}
\noindent The following lemma is the main ingredient of the proof of the decorrelation estimate as well as the heart of the present paper: 
\begin{lemma}\label {L:mainlemma}
Let $E\ne E'$ be two positive energies in the localized regime and $\beta \in (1/2, 1)$. Assume that $\Lambda=\Lambda_L=:[-L,L]$ is a large interval in $\mathbb Z.$
Pick $c_1, c_2>0$ and denote by $\mathbb P^{*}$ the probability of the following event (called (*)): \\
There exists two simple eigenvalues of $H_{\omega} (\Lambda)$, say $E(\omega), E'(\omega)$ such that $|E(\omega)-E|+|E'(\omega)-E'|\leqslant
e^{-L^\beta}$ and 
$$ \| \nabla_{\omega}(c_1 E(\omega)-c_2 E'(\omega)) \|_1\leqslant c_1 e^{-L^\beta} .$$
Then, there exists $c>0$ such that
$$ \mathbb P^{*}\leqslant e^{-cL^{2\beta}} .$$ 
\end{lemma}
\begin{remark}\label{R:}
{ \rm There is a slightly 
difference between the above lemma and Lemma 2.4 in \cite{FK2011} where $c_1=c_2=1$.
\noindent In fact, in the proof of Theorem \ref{T:decorrelation}, we will use the above lemma with $c_1, c_2$ are respectively $\frac{1}{E}, \frac{1}{E'}$ which are two distinct, positive numbers. This difference results from the lack of the normalization of $\|\nabla E(\omega)\|_1$ for our model.
\noindent Moreover, we will see in Remark \ref{R:dim} at the end of this section that, for the model (\ref{eq:1.1}), if $c_1=c_2,$ $P^{*}$ is equal to $0$ for all $L$ large.}
\end{remark}
\noindent We will skip for a moment the proof of Lemma \ref{L:mainlemma} and recall how to use the above lemma to complete the proof of Theorem \ref{T:decorrelation}. This part can be found in \cite{FK2011}. We repeat it here with tiny but necessary  changes adapted for the model (\ref{eq:1.1}).
\begin{proof}[Proof of Theorem \ref{T:decorrelation}]
Recall that $E, E'$ are two positive, fixed energies and $J_L= E+L^{-1}[-1,1]$ and $J'_{L} = E'+L^{-1}[-1,1]$. One chooses $L$ large enough such that  $\min \{E-L^{-1}, E'-L^{-1}\} \geq \min \{E, E'\}/2 >0$.\\ 
Let $cL^{\alpha} \leqslant l \leqslant L^{\alpha}/c$ with $c>0$. From the Minami estimate, one has
\begin{align*}
&\mathbb P \big(\# \{\sigma(H_{\omega}(\Lambda_l) ) \cap J_L\} \geqslant 2 \text { or } \# \{\sigma(H_{\omega}(\Lambda_l) ) \cap J'_L\} \geqslant 2 \big  )\\
& \leqslant  \dfrac{4 \beta_0 \|\rho\|_{\infty}\|s\rho(s)\|_{\infty}}{(\min\{E,E'\})^2} (|\Lambda_l||J_L|)^{2}\leq C (l/L)^{2}
\end{align*}
where $C$ is a constant depending only on $E, E', \beta_0$ and $\rho$.\\
Hence, it is sufficient to prove that $\mathbb P_0 \leqslant C(l/L)^{2}  e^{(\log L)^{\beta}}$ where 
\begin{equation} \label{eq:3.20}
\mathbb P_0:=   \mathbb P \big(\# \{\sigma(H_{\omega}(\Lambda_l) ) \cap J_L\} = 1; \; \# \{\sigma(H_{\omega}(\Lambda_l) ) \cap J'_L\} =1 \big  ).
\end{equation}
The crucial idea of proving decorrelation estimates in \cite{FK2011} is to reduce the proof of (\ref{eq:3.20}) to the proof of a similar estimate where $\Lambda_l$ is replaced by a much smaller cube, a cube of side length of order $\log L$. 
Precisely, one has (c.f. Lemma 2.1 and Lemma 2.2 in \cite{FK2011}):    
$$ \mathbb P_0 \leqslant C (l/L)^{2} + C (l/\widetilde{l}) \mathbb P_1 $$
where $\widetilde{l} = C \log L$ and 
$$ \mathbb P_1: =   \mathbb P \big(\# [\sigma(H_{\omega}(\Lambda_{\widetilde{l}}) ) \cap \widetilde{J}_L] \geqslant 1 \text{ and }  \# [\sigma(H_{\omega}(\Lambda_{\widetilde {l}}) ) \cap \widetilde{J}'_L
] \geqslant 1 \big  )$$ 
where $\widetilde{J}_L =E+L^{-1}(-2,2)$ \text{ and } $\widetilde{J}'_L
 =E'+L^{-1}(-2,2).$\\
To complete the proof of Theorem \ref{T:decorrelation}, one need show that 
\begin{equation} \label{eq:3.16}
\mathbb P_1\leqslant C (\widetilde{l}/L)^{2} e^{\widetilde{l}^{\beta}}.
\end{equation}
Thanks to the Minami estimate and the following inequality (see Lemma 2.3 in \cite{FK2011} for a proof)
\begin{equation}\label{eq:j}
\| Hess_{\omega}(E(\omega))\|_{l^{\infty}\rightarrow l^1} \leqslant \dfrac{C}{ dist \big( E(\omega), \sigma(H_{\omega}(\Lambda)) \setminus E(\omega)  \big)},
\end{equation}
\noindent one infers that
\begin{equation*}
\mathbb P \Big(
\begin{Bmatrix}
\sigma (H_{\omega}(\Lambda_{\widetilde{l}}))\cap (\widetilde{J}_L) = \{E(\omega)\}\\ 
\| Hess_{\omega}(E(\omega))  \|_{l^{\infty}\rightarrow l^1} \geqslant \epsilon^{-1} 
\end{Bmatrix}
\Big ) \leqslant C \epsilon \widetilde{l}^{2} L^{-1}.
\end{equation*}
Hence, for $\epsilon \in (4L^{-1},1),$ one has 
\begin{equation}\label{eq:3.17}
\mathbb P_1 \leqslant C \epsilon \widetilde {l}^{2}L^{-1}+\mathbb P_{\epsilon} 
\end{equation}
where $\mathbb P_{\epsilon} =\mathbb P(\Omega_0(\epsilon))$ with 
\begin{equation*}
\Omega_0(\epsilon): =
\begin{Bmatrix}
\sigma (H_{\omega}(\Lambda_{\widetilde {l}}))\cap \widetilde {J}_L = \{E(\omega)\}\\ 
\{E(\omega)\}= \sigma (H_{\omega}(\Lambda_{\widetilde {l}})) \cap (E-C\epsilon, E+C\epsilon) \\
\sigma (H_{\omega}(\Lambda_{\widetilde {l}}))\cap \widetilde{J}'_L = \{E'(\omega)\}\\ 
\{E'(\omega)\}= \sigma (H_{\omega}(\Lambda_{\widetilde {l}})) \cap (E'-C\epsilon, E'+C\epsilon) \\
\end{Bmatrix}
.
\end{equation*}
Next, one puts $\lambda:=e^{-\widetilde {l}^{\beta}}$ and defines, for $\gamma, \gamma' \in \Lambda_{\widetilde{l}},$ 
\begin{align*}
\Omega_{0,\beta}^{\gamma,\gamma'}(\epsilon): =
\Omega_0(\epsilon) \cap \{ \omega\mid |J_{\gamma,\gamma'}(E(\omega),E'(\omega))| \geqslant \lambda  \} 
 \end{align*}
 where $J_{\gamma,\gamma'}(E(\omega),E'(\omega))$ is the Jacobian of the mapping 
 $$(\omega_{\gamma},\omega_{\gamma'}) \mapsto (E(\omega),E'(\omega)).$$
On the one hand, $\mathbb P_{\epsilon}$  can be dominated as follows: 
$$ \mathbb P_{\epsilon} \leqslant \sum_{\gamma \ne \gamma'} \mathbb P (\Omega_{0,\beta}^{\gamma,\gamma'}(\epsilon)) +\mathbb P_r$$
where $\mathbb P_r$ is the probability of the following event 
$$ \mathcal D: = \{\omega\in \Omega_0(\epsilon) \mid |J_{\gamma,\gamma'}(E(\omega),E'(\omega))| \leqslant \lambda \text { for all } \gamma, \gamma'\in \Lambda _{\widetilde {l}}\}. $$
On the other hand, from Lemma 2.6 in \cite{FK2011}, it is known that  
$$ \mathbb P (\Omega_{0,\beta}^{\gamma,\gamma'}(\epsilon)) \leqslant C L^{-2}\lambda^{-4} \; \text{ for all } \gamma, \gamma' \in \Lambda _{\widetilde {l}}  .$$
Hence, 
\begin{equation}\label{eq:3.18}
\mathbb P_{\epsilon} \leqslant C \widetilde{l}^2 L^{-2}\lambda^{-4} +\mathbb P_r.
\end{equation}
Choose $\epsilon:=L^{-1}\lambda^{-3}$, (\ref{eq:3.17}) and (\ref{eq:3.18}) yield that
\begin{equation}\label{eq:3.19}
\mathbb P_1 \leqslant C  (\widetilde{l}/L)^{2} e^{\widetilde{l}^{\beta}} +\mathbb P_r.
\end{equation} 
Finally, we will use Lemma \ref{L:mainlemma} to estimate $\mathbb P_r$.\\
For each $\omega\in \mathcal D,$ we rewrite the Jacobian  $J_{\gamma,\gamma'}(E(\omega),E'(\omega))$ as follows:    
\begin{align}\label{eq:matrix}
& J_{\gamma,\gamma'}(E(\omega),E'(\omega)) =
\begin{vmatrix}
\partial_{\omega_{\gamma}} E(\omega)&\partial_{\omega_{\gamma'}} E(\omega)\\ 
\partial_{\omega_{\gamma}} E'(\omega)&\partial_{\omega_{\gamma'}} E'(\omega)\\ 
\end{vmatrix}\notag\\
&= \dfrac{E(\omega)E'(\omega)}{\omega_{\gamma} \omega_{\gamma'}} 
\begin{vmatrix}
\dfrac{1}{E(\omega)} \omega_{\gamma}\partial_{\omega_{\gamma}} E(\omega)& \dfrac{1}{E(\omega)}  \omega_{\gamma'}\partial_{\omega_{\gamma'}} E(\omega)\\ 
\dfrac{1}{E'(\omega)}\omega_{\gamma}\partial_{\omega_{\gamma}} E'(\omega)& \dfrac{1}{E'(\omega)} \omega_{\gamma'}\partial_{\omega_{\gamma'}} E'(\omega)\\ 
\end{vmatrix}
\end{align}
if $E(\omega)$ and $E'(\omega)$ are non-zero.
\vskip 1pt
\noindent Note that, from (\ref{eq:**}), one has
$$\sum_{\gamma\in \Lambda _{\widetilde {l}}}\dfrac{1}{E(\omega)} \omega_{\gamma}\partial_{\omega_{\gamma}}E(\omega)=\sum_{\gamma\in \Lambda _{\widetilde {l}}}\dfrac{1}{E'(\omega)} \omega_{\gamma}\partial_{\omega_{\gamma}}E'(\omega)=1.$$
Hence, one can apply Lemma 2.5 in \cite{FK2011} to (\ref{eq:matrix}) and deduce that
$$ \| \nabla_{\omega}(\dfrac{1}{E} E(\omega)- \dfrac{1}{E'} E'(\omega)) \|_1\leqslant e^{-\widetilde {l}^{\beta'}}$$
for any $1/2 <\beta'<\beta$.\\
Thus, Lemma ~\ref{L:mainlemma} yields that, for $L$ sufficiently large, 
\begin{equation}\label{eq:final}
 \mathbb P_r \leqslant  \widetilde{l}^2 e^{-c \widetilde{l}^{2 \beta'}} = O(L^{-\infty}).
\end{equation}
From (\ref{eq:3.19}) and (\ref{eq:final}), (\ref{eq:3.16}) follows and we have Theorem \ref{T:decorrelation} proved.
\end{proof}
\noindent Before coming to the proof of Lemma \ref{L:mainlemma}, we state and prove here a short lemma which will be used repeatedly in the rest of this section.
\begin{lemma}\label{L:det}
Pick $A\in Mat_{n}(\mathbb R)$ and $b \in \mathbb R^n$ such that $\|b\| \leq c_0e^{-L^{\beta}/2}$ where $c_0, L$ are fixed, positive constants. Assume that the following system of linear equations
$$Ax=b$$
has a solution $u$ satisfying $\|u\| \geq e^{-L^{\beta}/4}$.\\
Then, 
$$|\det A| \leq c_0\max \{1, ||\text{adj}(A)||\}e^{-L^{\beta}/4}$$
where $\text{adj}(A)$ is the adjugate of the matrix $A$.
\end{lemma}
\begin{proof}[Proof of Lemma \ref{L:det}]
Assume by contradiction that 
$$|\det A| > c_0 \max \{1, ||\text{adj}(A)||\} e^{-L^{\beta}/4}>0.$$ 
As a result, $A$ is invertible. Hence, $u=A^{-1}b$ is the unique solution of the system $Ax=b.$ We therefore infer that 
\begin{align*}
\|u\| \leq \|A^{-1}\| \|b\| = \dfrac{1}{|\det A|} \|\text{adj}(A)\|\|b\|\leq \dfrac{\max \{1, ||\text{adj}(A)||\}}{|\det A|}\|b\|. 
\end{align*}  
Hence, 
$$|\det A| \leq c_0 \max \{1, ||\text{adj}(A)||\} e^{-L^{\beta}/2} e^{L^{\beta}/4}=c_0 \max \{1, ||\text{adj}(A)||\} e^{-L^{\beta}/4}$$
which is a contradiction. 
\end{proof}
\noindent To complete the present section, we state here the proof of Lemma \ref{L:mainlemma}.
\vskip 2pt
\begin{proof}[Proof of Lemma \ref{L:mainlemma}]
Let $u:=u(\omega)$ and $v:=v(\omega)$ be normalized eigenvectors associated to $E(\omega)$ and $E'(\omega).$ By Lemma \ref{L:perturbation}, we have
$$ \nabla_{\omega}E(\omega)=\big(2 \|\Pi_{\gamma}u\|^2\big)_{\gamma\in \Lambda} \text{   and   } \nabla_{\omega}E'(\omega)=\big(2\|\Pi_{\gamma}v\|^2\big)_{\gamma\in \Lambda}  $$
We introduce the linear operator $T$ from $l^2(\Lambda)$ to $l^2(\Lambda)$ defined as follows
 $$ Tu(n)=u(n)-u(n+1) $$
where $u=(u(n))_{n} \in l^2(\Lambda).$ Recall that $\Lambda= \Lambda_L=\mathbb Z/ L\mathbb Z,$ i.e. we use periodic boundary conditions here.\\ 
Assume that $\{\omega\}_{j\in \Lambda}$ belongs to the event (*). We thus have
\begin{align*}
c_1e^{-L^\beta}&\geqslant \| \nabla_{\omega}(c_1E(\omega)-c_2E'(\omega)) \|_1 =\sum_{n}|(\sqrt{c_1}Tu(n))^2- (\sqrt{c_2}Tv(n))^2| \\
 & =\sum_{n}|\sqrt{c_1}Tu(n)-\sqrt{c_2}Tv(n)||\sqrt{c_1}Tu(n)+\sqrt{c_2}Tv(n)| .
\end{align*} 
Hence, $e^{-L^\beta} \geqslant \sum_{n}|Tu(n)-cTv(n)||Tu(n)+cTv(n)| $ with $c=\sqrt{c_2}/\sqrt{c_1}>0.$\\
Then, there exists a partition of $\Lambda=\{-L,\ldots,L\},$ say $\mathcal{P}\subset \Lambda$ and $\mathcal{Q}\subset \Lambda$ 
such that $\mathcal{P} \cup \mathcal{Q}  =\Lambda$,  $\mathcal{P} \cap \mathcal{Q} =\emptyset $ and  
\begin{itemize}
\item 
$\text{ for } n \in \mathcal P, |Tu(n)-cTv(n)| \leqslant e^{-L^{\beta}/2}$
\item 
$ \text { for } n \in \mathcal Q,$  $ |Tu(n)+cTv(n)|\leqslant e^{-L^{\beta}/2}$
\end{itemize}
From now on, we put $v(n):=cv(n).$ This abuse of notation changes nothing thanks to the linearity of the operator $T$.\\
Hence, we obtain that 
\begin{equation} \label{eq:heart}
\begin{cases}
|Tu(n) =Tv(n)+ \text{O}(e^{-L^{\beta}/2}) & \text{ if } n \in \mathcal P\\
|Tu(n)=-Tv(n) +\text{O}(e^{-L^{\beta}/2})& \text { if } n \in \mathcal Q 
\end{cases}
\end{equation}
From Lemma \ref{L:lowerbound}, there exists  $c_3>0$ depending only on $\alpha_0, \beta_0$ and an interval $J$ of the length $c_3L^{\beta}$ such that
\begin{equation}\label{eq:*}
|u(k)|^2+|u(k+1)|^2 \geqslant 2 e^{-L^{\beta}/2}
\end{equation}
for all $k\in J.$\\ 
Now, we decompose 
\begin{equation} \label{eq:decompose}
\mathcal{P} \cap J=\cup \mathcal P_j  \text { and } \mathcal{Q} \cap J=\cup \mathcal Q_j
\end{equation}
where $\mathcal P_j \text { and } \mathcal Q_j$ are intervals in $\mathbb Z.$\\
We will divide the rest of the proof into some lemmata. First of all, in the Lemma \ref{L:length}, we show a restriction on the length of each interval $\mathcal P_j$ and $\mathcal Q_j$ in $\mathbb Z.$ We will make use of this lemma later to prove a "reduction" lemma (Lemma \ref{L:reduce}).
Next, In Lemma \ref{L:system}, with any four consecutive points in $J,$ we explain how to form a (inhomogeneous) $10\times 10$ system of linear equations from (\ref{eq:heart}) and eigenequations for $u$ and $v$. 
Finally, we show some restrictions  on $\{\omega_j\}_{j\in \Lambda}$ in Lemma \ref{L:condition}. Thanks to this lemma and Lemma \ref{L:reduce}, Lemma \ref{L:mainlemma} follows.
\begin{lemma}\label{L:length}
Assume that $\{\omega_j\}_{j\in \Lambda}$ belongs to the event (*) defined in Lemma \ref{L:mainlemma}:
$H_{\omega} (\Lambda)$ has two simple eigenvalues $E(\omega), E'(\omega)$ such that $|E(\omega)-E|+|E'(\omega)-E'|\leqslant
e^{-L^\beta}$ and 
$$ \| \nabla_{\omega}(c_1 E(\omega)-c_2 E'(\omega)) \|_1\leqslant c_1 e^{-L^\beta} .$$
Denote by $u, v$ normalized eigenvectors associated to $E(\omega), E'(\omega)$ respectively and 
consider the decomposition $\{\mathcal P_i, \mathcal Q_j\}$ in (\ref{eq:decompose}). Then, any $\mathcal P_j$ or $\mathcal Q_j$ can not contain more than four points. 
\end{lemma}
\begin{proof}[Proof of Lemma \ref{L:length}]
Thanks to the equivalent role of $\mathcal P$ and $\mathcal Q,$ it is sufficient to prove Lemma \ref{L:length} for $\{\mathcal P_j\}_j$.\\
Assume by contradiction that there exists an interval $\mathcal P_j$ contain at least five consecutive points, say
$\mathcal P_j=\{n-2, n-1,n, n+1, n+2,\ldots, m\}$ with $m\geqslant n+2.$\\
First of all, thanks to (\ref{eq:heart}), we have
\begin{align}
Tu(n-2)&=Tv(n-2)+O(e^{-L^{\beta}/2}).\notag\\
Tu(n-1)&=Tv(n-1)+O(e^{-L^{\beta}/2}).\notag\\ \label{eq:3.5} 
Tu(n)&=Tv(n) +O(e^{-L^{\beta}/2}).\\
Tu(n+1)&=Tv(n+1)+O(e^{-L^{\beta}/2}).\notag \\
Tu(n+2)&=Tv(n+2)+O(e^{-L^{\beta}/2}).\notag
\end{align}
Next, consider the triple of consecutive points $\{n-2, n-1, n\} \in \mathcal P_j.$ Using the eigenequations for $u$ and $v$ at the point $(n-1)$ and take the hypothesis 
$|E(\omega)-E|+|E'(\omega)-E'|\leq e^{-L^{\beta}}$ into account, we deduce  
\begin{equation}\label{eq:3.6}
Eu(n-1)= \omega_{n-1} Tu(n-1) - \omega_{n-2}Tu(n-2)+O(e^{-L^{\beta}/2}),
\end{equation}
\begin{equation}\label{eq:3.30}
 E'v(n-1)= \omega_{n-1} Tv(n-1) - \omega_{n-2}Tv(n-2)+O(e^{-L^{\beta}/2}).
 \end{equation}
 Hence, (\ref{eq:3.6}), (\ref{eq:3.30}) and the first two equations in (\ref{eq:3.5}) yield  
\begin{equation}\label{eq:3.7}
Eu(n-1) = E' v(n-1) + O(e^{-L^{\beta}/2}).
\end{equation}
Similarly, we have
\begin{equation}\label{eq:3.8} 
Eu(n)=E'v(n)+O(e^{-L^{\beta}/2}).\\ 
\end{equation}
Combining (\ref{eq:3.7}), (\ref{eq:3.8}) and the second equation in (\ref{eq:3.5}), we obtain
$$\Big(1-\dfrac{E}{E'}\Big) u(n) = \Big(1-\dfrac{E}{E'}\Big) u(n-1)+O(e^{-L^{\beta}/2})$$ 
which implies that 
\begin{equation}\label{eq:3.10}
Tu(n-1) \leq Ce^{-L^{\beta}/2}
\end{equation}
where $C$ is a positive constant depending only on $E, E', \alpha_0$ and $\beta_0$. \\
\noindent Repeating again the above argument for other triples of consecutive points in $\mathcal P_j$, we obtain 
\begin{equation}\label{eq:3.11}
Tu(n) \leq Ce^{-L^{\beta}/2},
\end{equation}
and 
\begin{equation}\label{eq:3.15}
Tu(n+1) \leq Ce^{-L^{\beta}/2}.
\end{equation}
\noindent On the other hand, we have the following eigenequations for $u$ at the point $n$ and $n+1$
\begin{align*}
Eu(n)= \omega_{n} Tu(n) - \omega_{n-1}Tu(n-1) +O(e^{-L^{\beta}/2}),\\
Eu(n+1)= \omega_{n+1} Tu(n+1) - \omega_{n}Tu(n)+O(e^{-L^{\beta}/2}).
\end{align*}
Hence, combining the above equations and (\ref{eq:3.10})-(\ref{eq:3.15}), we infer that there exists a positive constant $C$ being independent of $L$ such that 
$$ |u(n)|^2 +|u(n+1)|^2 \leqslant Ce^{-L^{\beta}}$$
which contradicts (\ref{eq:*}) if we choose $L$ large enough.\\ 
Hence, an interval $\mathcal P_j$ or $\mathcal Q_j$ can not contain more than four points in $\mathbb Z $ and 
we have Lemma \ref{L:length} proved.\\ 
\end{proof}
\noindent From the proof of Lemma \ref{L:length}, we reach to the following conclusion: 
\begin{remark}\label{R:proportion}
{\rm If two consecutive points ordered increasingly belong to some interval $\mathcal P_j$ ($n-2,n-1$ for instance), the value of $u$ at the latter point ($n-1$ in this case) is proportional to the value of $v$ at that point (as in (\ref{eq:3.7})) up to an exponentially small error. Moreover, if we have three consecutive points ordered increasingly in some interval $\mathcal P_j$ ($n-2,n-1,n$), the middle point ($n-1$) always satisfies an inequality of the form (\ref{eq:3.10}). Finally, if three points $n-2,n-1,n$ belong to some $\mathcal Q_i,$ we will have almost the same conclusion except that $E'$ need replacing by $-E'$ in  (\ref{eq:3.7}) and (\ref{eq:3.8}).}
\end{remark}     
\begin{lemma}\label{L:system}
Let $J$ be the subinterval of $\Lambda$ where (\ref{eq:*}) holds and $n-2, n-1,n,n+1$ be four consecutive points in $J$. Assume the same hypotheses as in Lemma \ref{L:length} and put $U:=(u(n-2),\ldots,u(n+2),v(n-2),\ldots,v(n+2))^t$. 
Then, from (\ref{eq:heart}) and eigenequations for $u$ and $v$, we can form a $10\times 10$ system of linear equations which admits $U$ as one of its solutions. 
\end{lemma}
\begin{proof} [Proof of Lemma \ref{L:system}]
From (\ref{eq:heart}), for each of these four points, we have an equation of the form 
\begin{equation} \label{eq:3.22}
Tu(k) =\pm Tv(k) + O(e^{-L^{\beta}/2}), \;k = \overline{n-2,n+1}
\end{equation}
where the choice of $(+)$ or $(-)$ sign depends on whether $k$ belongs to $\mathcal P$ or $\mathcal Q$.  
So, we have $4$ (inhomogeneous) linear equations in hand.\\
On the other hand, we have $6$ eigenequations of eigenvectors $u$ and $v$ at the points $n-1, n$ and $n+1.$ Hence, apparently, we have $10$ linear equations corresponding to $10$ variables $\{u(n-2),\ldots,u(n+2),v(n-2),\ldots,v(n+2)\}.$ However, there is a couple of things here which should be made clearer.\\ 
First of all, to form our systems of linear equations, we use the following three eigenequations w.r.t. $u$ 
\begin{equation} \label{eq:3.25}
Eu(k)= \omega_{k} Tu(k) - \omega_{k-1}Tu(k-1)+O(e^{-L^{\beta}/2})
\end{equation}
where $k= \overline{n-1,n+1}.$\\
Next, we consider the eigenequations of $v$ at $k= \overline{n-1,n+1}$ 
 \begin{equation} \label{eq:3.23}
 E'v(k)= \omega_{k} Tv(k) - \omega_{k-1}Tv(k-1)+O(e^{-L^{\beta}/2}).
\end{equation}
Instead of using directly (\ref{eq:3.23}) for our $10\times 10$ systems of linear equations, we substitute (\ref{eq:3.22})
into the right hand side of (\ref{eq:3.23}) to get 
\begin{equation} \label{eq:3.24} 
 E'v(k)= \pm \omega_{k} Tu(k) \mp \omega_{k-1}Tu(k-1)+O(e^{-L^{\beta}/2})
\end{equation}
and (\ref{eq:3.24}) will be used for our $10\times 10$ systems of linear equations.\\
In Lemma \ref{L:condition}, we will write down these $10\times 10$ systems of linear equations as follows: 
The first four equations come from (\ref{eq:3.22}). Then, we write down the equations in (\ref{eq:3.25}) and (\ref{eq:3.24}). The fifth equation is (\ref{eq:3.25}) and the sixth is (\ref{eq:3.24}) with $k=n$. The seventh is (\ref{eq:3.25}) and the eighth is (\ref{eq:3.24}) with $k=n+1$. Lastly, (\ref{eq:3.25}) and (\ref{eq:3.24}) with $k=n-1$ are the ninth and the tenth equation in turn.\\
Finally, we make an important remark in case there exists an interval $\mathcal P_j$ or $\mathcal Q_j$ contains
at least two consecutive points of these four points, says  $j-1, j$. According to Remark \ref{R:proportion}, we have 
\begin{equation} \label{eq:3.26}
v(j) = \pm \dfrac{E}{E'}u(j)+O(e^{-L^{\beta}/2})  
\end{equation}
where (\ref{eq:3.26}) takes $(+)$ sign iff $j-1$ and $j\in \mathcal P.$\\ 
Whenever (\ref{eq:3.26}) holds true, we will use it to replace (\ref{eq:3.24}) w.r.t. $k=j$ in our systems of linear equations. This replacement simplifies these $10 \times 10$ systems of linear equations and makes them easier to analyze.
\end{proof}

\begin{definition}
A point $n \in J$ is an interior point of $J$ if the interval $[n-2, n+2]$ belongs to $J$.
\end{definition}
\noindent Let $n-2, n-1,n,n+1$ be interior points in $J$. 
We consider all possible $10 \times 10$ systems of linear equations which we can get from these points as in Lemma \ref{L:system}. We have four points, each of them can belong to $\mathcal P$ or $\mathcal Q.$ Hence, the number of choices for four points' belonging to $\mathcal P$ or $\mathcal Q$ equals $2^4=16$ which is also the total number of $10\times 10$ systems of linear equations obtained in Lemma \ref{L:system}. 
Furthermore, we have the following useful observation:
\begin{lemma} \label{L:reduce}
Assume the same hypotheses as in Lemma \ref{L:length}. 
Let $n-2, n-1,n,n+1$ be interior points in $J$ and $\{\mathcal P_i, \mathcal Q_j\}$ be the decomposition in (\ref{eq:decompose}). Then, we will only need to analyze $10\times 10$ systems of linear equations corresponding to the following four cases:
\begin{itemize}
\item[] \textit{First case: $n-2, n-1,n \in \mathcal P_j$ and $n+1\in \mathcal Q_j,$}\\
\begin{tikzpicture}
\draw[dotted] (-1,0) -- (4,0);
\draw [black, fill=black] (0,0) circle (0.1);
\draw (0, -0.1) node[below] {$n-2$};
\draw [black, fill=black] (1,0) circle (0.1);
\draw (1, -0.1) node[below] {$n-1$};
\draw [black, fill=black] (2,0) circle (0.1);
\draw (2, -0.17) node[below] {$n$};
\draw [black] (3,0) circle (0.1);
\draw (3, -0.1) node[below] {$n+1$};
\end{tikzpicture}
\item[] \textit{Second case: $n-2,n-1 \in \mathcal Q_j$ and $n,n+1\in \mathcal P_j,$}\\
\begin{tikzpicture}
\draw[dotted] (-1,0) -- (4,0);
\draw [black] (0,0) circle (0.1);
\draw (0, -0.1) node[below] {$n-2$};
\draw [black] (1,0) circle (0.1);
\draw (1, -0.1) node[below] {$n-1$};
\draw [black, fill=black] (2,0) circle (0.1);
\draw (2, -0.17) node[below] {$n$};
\draw [black, fill=black] (3,0) circle (0.1);
\draw (3, -0.1) node[below] {$n+1$};
\end{tikzpicture}
\item[] \textit{Third case: $n-2,n-1 \in \mathcal Q_j,$ $n\in \mathcal P_j$ and $n+1 \in \mathcal Q_{j+1},$}\\
\begin{tikzpicture}
\draw[dotted] (-1,0) -- (4,0);
\draw [black] (0,0) circle (0.1);
\draw (0, -0.1) node[below] {$n-2$};
\draw [black] (1,0) circle (0.1);
\draw (1, -0.1) node[below] {$n-1$};
\draw [black, fill=black] (2,0) circle (0.1);
\draw (2, -0.17) node[below] {$n$};
\draw [black] (3,0) circle (0.1);
\draw (3, -0.1) node[below] {$n+1$};
\end{tikzpicture}
\item[] \textit{Forth case: $n-2\in \mathcal Q_j, \; n-1 \in \mathcal P_j, \;  n\in \mathcal Q_{j+1} \;$ and $n+1 \in \mathcal P_{j+1}.$}\\
\begin{tikzpicture}
\draw[dotted] (-1,0) -- (4,0);
\draw [black] (0,0) circle (0.1);
\draw (0, -0.1) node[below] {$n-2$};
\draw [black, fill=black] (1,0) circle (0.1);
\draw (1, -0.1) node[below] {$n-1$};
\draw [black] (2,0) circle (0.1);
\draw (2, -0.17) node[below] {$n$};
\draw [black, fill=black] (3,0) circle (0.1);
\draw (3, -0.1) node[below] {$n+1$};
\end{tikzpicture}
\end{itemize} 
\end{lemma}
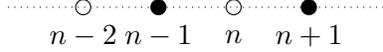
\begin{proof}[Proof of Lemma \ref{L:reduce}]
As mentioned above, we have a total of $16$ systems of linear equations to analyze. Thanks to the equivalent role of $\mathcal P$ and $\mathcal Q,$ we only have to consider a half of them. Apart from four cases listed above, the other cases could be:\\
 \textit{Fifth case:} {\it Assume that  all of these four points belong to some $\mathcal P_j.$}\\ 
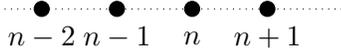
\begin{tikzpicture}
\draw[dotted] (-1,0) -- (4,0);
\draw [black, fill=black] (0,0) circle (0.1);
\draw (0, -0.1) node[below] {$n-2$};
\draw [black, fill=black] (1,0) circle (0.1);
\draw (1, -0.1) node[below] {$n-1$};
\draw [black, fill=black] (2,0) circle (0.1);
\draw (2, -0.17) node[below] {$n$};
\draw [black, fill=black] (3,0) circle (0.1);
\draw (3, -0.1) node[below] {$n+1$};
\end{tikzpicture}

\noindent Hence, From Lemma \ref{L:length}, $n+2 \in \mathcal Q_j.$ We consider four points $n-1,n, n+1, n+2$ and come back to \textit{First case}.\\ 
\noindent \textit{Sixth case:} {\it Suppose that $n-2\in \mathcal Q_j, \; n-1,n \in \mathcal P_j, \;$ and $n+1 \in \mathcal Q_{j+1}.$}\\
\begin{tikzpicture}
\draw[dotted] (-1,0) -- (4,0);
\draw [black] (0,0) circle (0.1);
\draw (0, -0.1) node[below] {$n-2$};
\draw [black, fill=black] (1,0) circle (0.1);
\draw (1, -0.1) node[below] {$n-1$};
\draw [black, fill=black] (2,0) circle (0.1);
\draw (2, -0.17) node[below] {$n$};
\draw [black] (3,0) circle (0.1);
\draw (3, -0.1) node[below] {$n+1$};
\end{tikzpicture}

\noindent We consider the point $n+2.$ 
If $n+2$ belongs to $\mathcal Q_{j+1}$, we consider four points $n-1,n,n+1,n+2$ and come back to \textit{Second case} because of the equivalent role of $\mathcal{P}$ and $\mathcal Q$. Otherwise, $n+2$ belongs to $\mathcal P_{j+1}.$ In this case, we consider four points $n-1,n,n+1,n+2$ and come back to \textit{Third case}.\\
\noindent \textit{Seventh case:} {\it Assume that $n-2\in \mathcal Q_j,$ $n-1 \in \mathcal P_j$ and $n, n+1 \in \mathcal Q_{j+1}.$}\\
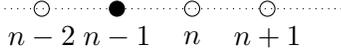
\begin{tikzpicture}
\draw[dotted] (-1,0) -- (4,0);
\draw [black] (0,0) circle (0.1);
\draw (0, -0.1) node[below] {$n-2$};
\draw [black, fill=black] (1,0) circle (0.1);
\draw (1, -0.1) node[below] {$n-1$};
\draw [black] (2,0) circle (0.1);
\draw (2, -0.17) node[below] {$n$};
\draw [black] (3,0) circle (0.1);
\draw (3, -0.1) node[below] {$n+1$};
\end{tikzpicture}

\noindent In this case, we consider four points $n-3, n-2, n-1, n.$ If $n-3$ also belongs to $\mathcal Q_j,$ we come back to \textit{Third case}. Otherwise, we come back to \textit{Forth case} on account of the equivalent role of $\mathcal P$ and $\mathcal Q$.\\
\noindent \textit{Eighth case:} {\it Suppose that $n-2 \in Q_j$ and $n-1, n, n+1\in \mathcal P_j$ for some $j.$}\\ 
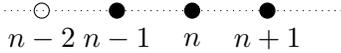
\begin{tikzpicture}
\draw[dotted] (-1,0) -- (4,0);
\draw [black] (0,0) circle (0.1);
\draw (0, -0.1) node[below] {$n-2$};
\draw [black, fill=black] (1,0) circle (0.1);
\draw (1, -0.1) node[below] {$n-1$};
\draw [black, fill=black] (2,0) circle (0.1);
\draw (2, -0.17) node[below] {$n$};
\draw [black, fill=black] (3,0) circle (0.1);
\draw (3, -0.1) node[below] {$n+1$};
\end{tikzpicture}

\noindent If $n+2 \in \mathcal Q_{j+1},$ we consider four points $n-1, n, n+1, n+2$ and come back to \textit{First case}.  
Otherwise, $n+2$ still belongs to $\mathcal P_j,$ hence $n+3 \in \mathcal Q_{j+1}$ according to Lemma \ref{L:length}. On the other hand, $n+3$ still belongs to $J$ since $n+1$ is the interior point of $J$. Hence, we consider four points $n, n+1, n+2, n+3$ and come back to \textit{First case}.\\ 
To conclude, we only need to analyze $4$ special cases. The other cases can be reduced to those ones.  
\end{proof}
\noindent Now, we come to the final stage in the proof of Lemma \ref{L:mainlemma} where we deduce the restrictions  on r.v.'s $\omega_j.$
\begin{lemma}\label{L:condition}
Assume hypotheses as in Lemma \ref{L:length}. Let $J$ be the interval defined in (\ref{eq:heart}) and $n-2, n-1, n, n+1$ be four interior points of $J$. Assume that these four points correspond to one of the four cases listed in Lemma \ref{L:reduce}.
Then, one of the following restrictions  on r.v.'s holds true
\begin{itemize}
\item[$(i)$] $\Big|\omega_n +\dfrac{E'-E}{4}\Big|\leqslant C e^{-L^{\beta}/8},$
\item[$(ii)$] $\Big|\omega_{n-1} +\dfrac{E'-E}{4}\Big|\leqslant C e^{-L^{\beta}/8},$
\item[$(iii)$] $\Big|\omega_{n-1}\omega_n -\dfrac{(E+E')^2}{4}\Big|\leqslant C e^{-L^{\beta}/4}.$
\end{itemize} 
\end{lemma} 
\begin{proof}[Proof of Lemma \ref{L:condition}]
For each of four cases in Lemma \ref{L:reduce}, we consider the corresponding system of linear equations formed in Lemma \ref{L:system} and compute its determinant. This yields some restrictions  on r.v.'s.\\ 
Recall that $U:=(u(n-2),\ldots, u(n+2), v(n-2), \ldots, v(n+2))^{t}.$\\ 
\textit{First case: Assume that three points $n-2, n-1,n \in \mathcal P_j$ and the other one $n+1\in \mathcal Q_j.$}\\
\noindent Since $n-2, n-1, n\in \mathcal P_j,$ two equations in (\ref{eq:3.24}) associated to $n-1, n$ will be replaced by
two equations of the type (\ref{eq:3.26}) with (+) sign in our system. Hence, according to Lemma \ref{L:system}, $U$
satisfies the following system of linear equations:
\begin{equation} \label{eq:A0}
A_0 U=b_0 
\end{equation}
where $b_0=(b_0^j)_{1\leqslant j\leqslant10}$ with 
$\|b_0\|\leqslant e^{-L^{\beta}/2}$
 and $A_{0}$ is the $10 \times 10$ matrix of the block form $(A_0^1|A_0^2)$ with
{\small
\begin{equation*}
A_0^1:=
\begin{pmatrix}
1&-1&0&0&0\\ 
0&1&-1&0&0\\ 
0&0&1&-1&0 \\
0&0&0&1&-1\\ 
0&-\omega_{n-1}&\omega_{n-1}+\omega_{n}-E&-\omega_n&0\\ 
0&0&\dfrac{E}{E'}&0&0\\ 
0&0&-\omega_n&\omega_n+\omega_{n+1}-E&-\omega_{n+1}\\ 
0&0&\omega_n&\omega_{n+1}-\omega_n&-\omega_{n+1}\\ 
-\omega_{n-2}&\omega_{n-2}+\omega_{n-1}-E&-\omega_{n-1}&0&0\\ 
0&\dfrac{E}{E'}&0&0&0\\
\end{pmatrix}
\end{equation*}  
}
and
\begin{equation*}
A_0^2:=
\begin{pmatrix}
1&-1&0&0&0\\
0&1&-1&0&0\\
0&0&1&-1&0\\
0&0&0&-1&1\\
0&0&0&0&0\\
0&0&-1&0&0\\
0&0&0&0&0\\
0&0&0&E'&0\\
0&0&0&0&0\\
0&-1&0&0&0\\
\end{pmatrix}.
\end{equation*}
\noindent Let $\text{adj}(A_0)$ be the adjugate of $A_0.$ It is easy to see that  
$$ \max\{1,\|\text{adj}(A_0)\|  \} \leqslant M_0$$
where $M_0$ is a positive constant depending only on $E, E', \alpha_0, \beta_0.$\\
Hence, thanks to Lemma \ref{L:det}, we have
$$|\det A_0| \leq M_0e^{-L^{\beta}/4}.$$
By an explicit computation in Appendix \ref{S:A}, we have
$$ |\det A_0| = \dfrac{4E}{E'}(E+E')\omega_{n-2}\omega_{n+1} \Big|\omega_{n} +\dfrac{E'-E}{4}\Big| \leqslant M_0 e^{-L^{\beta}/4}.$$ 
Therefore, from the fact that $E, E'>0$ and $\omega_j \geq \alpha_0>0$, the following condition on $\omega$ holds true with
$L$ sufficiently large:  
$$\Big|\omega_n +\dfrac{E'-E}{4}\Big|\leqslant C e^{-L^{\beta}/4}. \eqno{(I)}$$
\noindent \textit{Second case: $n-2,n-1 \in \mathcal Q_j$ and $n,n+1\in \mathcal P_j.$}\\
\noindent In the present case, since $n-2, n-1 \in \mathcal Q_j,$ we use (\ref{eq:3.26}) with (-) sign w.r.t. $n-1$ to replace the equation in (\ref{eq:3.24}) w.r.t. $n-1$ for our system of linear equations. Besides, since $n, n+1 \in \mathcal P_j,$ the equation in  (\ref{eq:3.24}) w.r.t. $n+1$ will be replaced by (\ref{eq:3.26}) with (+) sign at $n+1$.\\
Hence, according to Lemma \ref{L:system}, we have the following $10\times 10$ system of linear equations:  
\begin{equation} \label{eq:A1}
A_{1}U=b_1 
\end{equation}
where $b_1=(b_1^j)_{1\leqslant j\leqslant10}$ with 
$\|b_1\|\leqslant 10 e^{-L^{\beta}/2}$
 and $A_{1}= (A_1^1|A_1^2)$ where
{\small
\begin{equation*}   
A_1^1:=
\begin{pmatrix}
1&-1&0&0&0\\
0&1&-1&0&0\\ 
0&0&1&-1&0\\ 
0&0&0&1&-1\\ 
0&-\omega_{n-1}&\omega_{n-1}+\omega_{n}-E&-\omega_n&0\\ 
0&\omega_{n-1}&\omega_{n}-\omega_{n-1}&-\omega_n&0\\ 
0&0&-\omega_n&\omega_n+\omega_{n+1}-E&-\omega_{n+1}\\ 
0&0&0&\dfrac{E}{E'}&0\\ 
-\omega_{n-2}&\omega_{n-2}+\omega_{n-1}-E&-\omega_{n-1}&0&0\\ 
0&-\dfrac{E}{E'}&0&0&0\\ 
\end{pmatrix}
\end{equation*}
}
and 
\begin{equation*}
A_1^2:=
\begin{pmatrix}
-1&1&0&0&0\\ 
0&-1&1&0&0\\
0&0&1&-1&0\\
0&0&0&1&-1\\
0&0&0&0&0\\
0&0&-E'&0&0\\
0&0&0&0&0\\
0&0&0&-1&0\\
0&0&0&0&0\\
0&-1&0&0&0\\
\end{pmatrix}.
\end{equation*}
 
\noindent Again, by using Lemma \ref{L:det}, we infer that 
$$|\det A_1| \leqslant M_1 e^{-L^{\beta}/4}$$
where $M_1=M_1(E,E',\alpha_0,\beta_0) >0.$\\ 
Compute the determinant of $A_1$ (See Appendix \ref{S:A}), we obtain 
$$ |\det A_1| = \dfrac{4E}{E'}\omega_{n-2}\omega_{n+1} \Big|\omega_{n-1}\omega_n -\dfrac{(E+E')^2}{4}\Big| \leqslant M_1 e^{-L^{\beta}/4}.$$ 
Hence, take $\omega_j \geqslant \alpha_0>0$ and $E, E' >0$ into account, we have   
$$\Big|\omega_{n-1}\omega_n -\dfrac{(E+E')^2}{4}\Big|\leqslant C e^{-L^{\beta}/4} \eqno{(II)} $$
as $L$ sufficiently large.\\
\noindent \textit{Third case: $n-2,n-1 \in \mathcal Q_j,$ $n\in \mathcal P_j$ and $n+1 \in \mathcal Q_{j+1}.$}\\
\noindent According to \textbf{Step 2} and Lemma \ref{L:det}, we have 
$$ A_{2} U =b_2 $$
where $A_2=(A_2^1|A_2^2)$ is the $10\times 10$ matrix defined by
{\small 
\begin{equation*}
A_2^1:=
\begin{pmatrix}
1&-1&0&0&0\\ 
0&1&-1&0&0\\ 
0&0&1&-1&0\\ 
0&0&0&1&-1\\ 
0&-\omega_{n-1}&\omega_{n-1}+\omega_{n}-E&-\omega_n&0\\ 
0&\omega_{n-1}&\omega_{n}-\omega_{n-1}&-\omega_n&0\\ 
0&0&-\omega_n&\omega_n+\omega_{n+1}-E&-\omega_{n+1}\\ 
0&0&-\omega_n&\omega_n-\omega_{n+1}&\omega_{n+1}\\ 
-\omega_{n-2}&\omega_{n-2}+\omega_{n-1}-E&-\omega_{n-1}&0&0\\ 
0&-\dfrac{E}{E'}&0&0&0\\ 
\end{pmatrix}
\end{equation*}
}
and
\begin{equation*}
A_2^2:=
\begin{pmatrix}
-1&1&0&0&0\\
0&-1&1&0&0\\
0&0&1&-1&0\\
0&0&0&-1&1\\
0&0&0&0&0\\
0&0&-E'&0&0\\
0&0&0&0&0\\
0&0&0&-E'&0\\
0&0&0&0&0\\
0&-1&0&0&0\\
\end{pmatrix}.
\end{equation*}
\noindent such that $|\det A_{2}| \leq  M_2 e^{-L^{\beta}/4}$ where $M_2$ is some positive constant.\\
Then, by an explicit computation, we obtain that
$$|\det A_{2}| = 4E (E+E') \omega_{n-2}\;\omega_{n+1}\Big|\omega_n + \dfrac{E'-E}{4}\Big| \leqslant M_2 e^{-L^{\beta}/4}$$
which immediately yields that   
$$\Big|\omega_{n} + \dfrac{E'-E}{4}\Big| \leqslant C e^{-L^{\beta}/4} \eqno{(III)}$$
as $L\rightarrow +\infty$.\\
\noindent \textit{Forth case: Suppose that $n-2\in \mathcal Q_j, \; n-1 \in \mathcal P_j, \;  n\in \mathcal Q_{j+1} \;$ and $n+1 \in \mathcal P_{j+1}.$}\\
\noindent In this case, $U$ satisfies the following system of linear equations:
$$ A_{3} U =b_3$$
where $A_3=(A_3^1|A_3^2)$ is the $10\times 10$ matrix defined by
{\small
\begin{equation*}
A_3^1:=
\begin{pmatrix}
1&-1&0&0&0\\ 
0&1&-1&0&0\\ 
0&0&1&-1&0\\ 
0&0&0&1&-1\\ 
0&-\omega_{n-1}&\omega_{n-1}+\omega_{n}-E&-\omega_n&0\\ 
0&-\omega_{n-1}&\omega_{n-1}-\omega_{n}&\omega_n&0\\ 
0&0&-\omega_n&\omega_n+\omega_{n+1}-E&-\omega_{n+1}\\ 
0&0&\omega_n&\omega_{n+1}-\omega_{n}&-\omega_{n+1}\\ 
-\omega_{n-2}&\omega_{n-2}+\omega_{n-1}-E&-\omega_{n-1}&0&0\\ 
\omega_{n-2}&\omega_{n-1}-\omega_{n-2}&-\omega_{n-1}&0&0\\ 
\end{pmatrix}
\end{equation*}
}
and 
\begin{equation*}
A_3^2:=
\begin{pmatrix}
-1&1&0&0&0\\
0&1&-1&0&0\\
0&0&-1&1&0\\
0&0&0&1&-1\\
0&0&0&0&0\\
0&0&-E'&0&0\\
0&0&0&0&0\\
0&0&0&-E'&0\\
0&0&0&0&0\\
0&-E'&0&0&0\\
\end{pmatrix}
\end{equation*}
\noindent and $|\det A_3| \leq M_3 e^{-L^{\beta}/4}.$\\
We compute 
$$|\det A_{3}| = EE'  \omega_{n-2}\times \omega_{n+1}\times \; \Big|4\omega_{n-1}+ E'-E\Big| \times \; \Big|E'-E +4\omega_n\Big|.$$ 
Hence, at least one of the two following conditions on $\omega$ must be satisfied:
\begin{itemize}
\item $\Big|\omega_{n-1}+ \dfrac{E'-E}{4}\Big| \leqslant C e^{-L^{\beta}/8}, \quad (IV)$
\item $\Big |\omega_n+ \dfrac{E'-E}{4}\Big| \leqslant C e^{-L^{\beta}/8}.$
\end{itemize}
From $(I)-(IV)$, Lemma \ref{L:condition} follows. 
\end{proof}
\noindent Lemmata \ref{L:reduce} and \ref{L:condition} yield that if we consider any $4$ consecutive interior points of $J,$ we obtain at least one condition of the types $(i)-(iii).$ Hence,   
the random variables $\{\omega_j\}_{j\in \Lambda}$ must satisfy at least $|J|/8=cL^{\beta}$ conditions of the types $(i)-(iii)$. From the fact that $\omega_n$ are i.i.d. and possess a bounded density, the conditions $(i)-(iii)$ imply that the event (*) defined in Lemma \ref{L:mainlemma} can occur for a given partition $\mathcal{P}$ and $\mathcal Q$ with a probability at most $e^{-cL^{2\beta}}$ for some $c>0.$ Hence, 
$$\mathbb P^{*}\leqslant 2^L e^{-cL^{2\beta}}\leqslant e^{-\widetilde{c}L^{2\beta}} \text{ with } 0<\widetilde{c}<c$$ 
as the number of partitions is bounded by $2^{L}$ and $\beta>1/2.$\\ 
We therefore have Lemma \ref{L:mainlemma} proved.

\end{proof}
\begin{remark} \label{R:dim}
{\rm
Thanks to the equality (\ref{eq:**}), it is not hard to derive the following estimate for the model (\ref{eq:1.1}), 
\begin{equation} \label{eq:3.21}
\dfrac{\Delta E}{2\beta_{0}} |\Lambda|^{-1/2} \leqslant  \| \nabla_{\omega}(E(\omega)-E'(\omega)) \| \leqslant \| \nabla_{\omega}(E(\omega)-E'(\omega)) \|_1  
 \end{equation}
 provided that $|E(\omega)-E|+|E'(\omega)-E'|\leqslant e^{-L^\beta}$
and $\Delta E=|E-E'|$. \\
The above estimate reads that the $l^1-$distance of the gradients of $E(\omega)$ and $E'(\omega)$ is bounded from below by a positive term that is polynomially small w.r.t. the length of the interval $\Lambda.$\\
Now, let $c_1= c_2$ in Lemma \ref{L:mainlemma}.\\
Under the hypotheses in Lemma \ref{L:mainlemma}, the estimate (\ref{eq:3.21}) implies that, for any $\{\omega_j\}_{j\in \Lambda} \text{ belonging to the event (*)}$, 
 $$\dfrac{\Delta E}{2\beta_{0}} |\Lambda|^{-1/2} \leqslant C e^{-c|\Lambda|^{\beta}} $$
which is impossible when $|\Lambda|$ sufficiently large. Hence, for $c_1=c_2,$ $\mathbb P^{*}$ is equal to $0$ .
\vskip 1pt
\noindent Finally, we would like to note that an estimate like (\ref{eq:3.21}) for the discrete Anderson model only holds true for two distinct energies sufficiently far apart from each other. Moreover, that kind of estimate enable us to prove the decorrelation estimate for the discrete Anderson model in \emph{any} dimension (c.f. Lemma 2.4 in \cite{FK2011}). But it is not the case for the model (\ref{eq:1.1}).}     
\end{remark}
\section{Comment on the lower bound of the r.v.'s} \label{S:remark}
In this section, we want to discuss how to relax the hypothesis on the lower bound of random variables $\{\omega_j\}_{j\in \mathbb Z}.$
\vskip 0.5 pt
\noindent Assume that all r.v.'s $\omega_j$ are only non-negative instead of being bounded from below by a positive constant. Precisely, assume that $\omega_j \in [0, \beta_0]$ $\forall j\in \mathbb Z.$ In order to carry out our proof, we have to assume an extra condition on the function of distribution $F(t)$ of r.v.'s $\omega_j:$
\begin{equation}\label{eq:4.1}
F(t):=\mathbb P (\omega_j \leqslant t) \leqslant e^{-t^{-\eta}}
\end{equation}
for all small positive $t$, where $\eta$ is some positive number.\\
Without loss of generality, let $\eta=1.$ The condition (\ref{eq:4.1}) means that the distribution $F(t)$ is exponentially small in a neighborhood of $0$.\\
Now, let $\Lambda=[-L,L]$ be an interval in $\mathbb Z,$ we have 
\begin{equation}\label{eq:4.2}
\mathbb P (\exists \omega_{\gamma} \leqslant e^{-(\log L)^{\delta}} \text{\;with\;} \gamma \in \Lambda) \leqslant (2L+1)e^{-e^{(\log L)^{\delta}}}. 
\end{equation}
Note that the right hand side of (\ref{eq:4.2}) converges to $0$ as $L \rightarrow \infty$. Here $\delta$ is a fixed number in $(0,1)$.\\
Hence, with a probability greater than or equal to $1-(2L+1) e^{-e^{(\log L)^{\delta}}},$\\ 
we have 
\begin{equation}\label{eq:4.3}
\omega_j \geq e^{-(\log L)^{\delta}}>0 \; \forall j \in[-L,L].
\end{equation}
We will use (\ref{eq:4.3}) to prove the following "lower bound" for normalized eigenvectors of $H_{\omega}(\Lambda).$
 \begin{lemma} \label{L:lowerbound1} 
Pick $\beta \in (1/2, 1)$ and a fixed number $\epsilon \in (0,\beta).$ 
Let $\Lambda=[-L,L]$ be a large cube in $\mathbb Z.$ Suppose that $E(\omega)$ is an eigenvalue of $H_{\omega}(\Lambda)$ and $u:=u(\omega)$ is its associated normalized eigenvector.\\
\noindent Then, with a probability greater than or equal to $1-(2L+1) e^{-e^{(\log L)^{\delta}}}$,\\ 
there exists a point $k_0$ in $\Lambda$ such that 
$$u^2(k)+u^2(k+1)\geqslant e^{-L^{\beta}/2}$$ 
for all $|k-k_0|\leqslant \frac{1}{4}L^{\beta-\epsilon}$ as $L$ large enough.
\end{lemma}
\begin{proof}[Proof of Lemma \ref{L:lowerbound1}] 
Consider $\{\omega_j\}_{j\in \Lambda}$ such that (\ref{eq:4.3}) holds true. Using the same notations and proceed as in Lemma \ref{L:lowerbound}, for $n, m \in \Lambda,$ we have
$$v(n) = T(n,E(\omega))\cdots T(n-m+1,E(\omega))v(m)$$
 where $T(n,E(\omega))$ and $v(n)$ are the transfer matrix and column vector defined in the proof of Lemma \ref{L:lowerbound}.
 Thanks to (\ref{eq:4.3}), the transfer matrices $T(n,E(\omega))$ are well defined and invertible. Moreover, they and their inverse matrices are bounded by $C_L:=e^{c(\log L)^{\delta}}$ where $c>0$ depends only on $\beta_0.$\\
Thus, 
\begin{equation}\label{eq:4.4}
\|v(n)\|\leqslant C_L^{|n-m|} \|v(m)\| =e^{ c(\log L)^{\delta}|n-m|}\|v(m)\|
\end{equation}
where $n, m \in \Lambda.$\\
Assume that $\|v(k_0)\|$ is the maximum of $\|v(n)\|.$ Hence, 
 \begin{equation} \label{eq:4.5}
 \|v(k_0)\| \geq \dfrac{1}{\sqrt{L}}
 \end{equation}
 as $u$ is a normalized vector.\\
Pick $\kappa>0$ a fixed number and consider integers $k$ such that $|k-k_0|\leqslant  \kappa L^{\beta-\epsilon}.$ From (\ref{eq:4.4}) and (\ref{eq:4.5}), we have the following inequality  
$$ \|v(k)\| \geqslant  \dfrac{1}{\sqrt{L}} e^{ -c(\log L)^{\delta}|k-k_0|}\geq \dfrac{1}{\sqrt{L}} e^{ -c \kappa (\log L)^{\delta}L^{\beta-\epsilon}}   \geq e^{-\kappa L^{\beta}} $$
when $L$ is sufficiently large.\\
Hence, by choosing $\kappa=1/4,$ we have 
$$u^2(k)  +u^2(k+1) \geqslant e^{-L^{\beta}/2}$$
which completes the Lemma \ref{L:lowerbound1}.
\end{proof}
\noindent Roughly speaking, we obtained almost the same "lower bound" for the normalized eigenvectors of finite volume operators, but with a good probability instead of the probability $1$ as in Lemma \ref{L:lowerbound}.\\ 
Now, let $\beta$ be a fixed number in the interval $(1/2,1).$\\ 
One the one hand, thanks to Lemma \ref{L:lowerbound1}, the argument in proof of Theorem \ref{L:mainlemma} still works out. In deed, in this case, we can proceed as in the proof of Lemma \ref{L:mainlemma} to obtain at least $cL^{\beta-\epsilon}$ restrictions  on r.v.'s $\omega.$ Hence, the upper bound for the probability $\mathbb P^{*}$ in Lemma \ref{L:mainlemma} is now 
 $$B:= 2^{L} (e^{-\widetilde{c}L^{\beta}})^{cL^{\beta-\epsilon}}= 2^{L}e^{-c L^{2\beta -\epsilon}}.$$ 
 If we choose $\epsilon$ in Lemma \ref{L:lowerbound1} small enough such that $2\beta -\epsilon>1$, the upper bound $B$ is exponentially small w.r.t. $L$. Hence, the Lemma \ref{L:mainlemma} still holds true.\\   
 On the other hand, we observe that the term $(2L+1)e^{-e^{(\log L)^{\delta}}}$ (the upper bound of the probability that (\ref{eq:4.3}) fails) is negligible compared to the $C (l/L)^{2}e^{(\log L)^{\beta}}$ (the right hand side of the decorrelation estimate in Theorem \ref{T:decorrelation}) as $L$ large. Hence, we obtain again the decorrelation estimate. In other words, Theorem \ref{T:decorrelation} and Theorem \ref{T:Poisson} still hold true in this case.
 
\section{More than two distinct energies}\label{S:morethan2}
In this section, we would like to show that, following an argument in \cite{FK2011}, we can use Theorem \ref{T:decorrelation} to prove the asymptotic independence for any fixed number of point processes. 
\begin{theorem} \label{T:Poisson1}
For a fixed number $n\geqslant 2,$ consider a finite sequence of fixed, positive energies $\{E_i\}_{1\leqslant i\leqslant n}$ in the localized regime such that $\nu(E_i)>0$ for all $1\leq i \leq n.$\\
Then, as $|\Lambda| \rightarrow +\infty,$ $n$ point processes $\Xi(\xi,E_i,\omega,\Lambda)$ defined as in (\ref{eq:1.3}) converge weakly to $n$ independent Poisson processes.
\end{theorem}
\begin{proof}[Proof of Theorem \ref{T:Poisson1}] 
We will prove in detail the case of $n=3$ with three distinct, positive energies $E, E', E''$.\\
Consider non-empty compact intervals $(U_j)_{1\leqslant j\leqslant J},$ $ (U'_{j})_{1\leqslant j\leqslant J'}, $$(U''_{j})_{1\leqslant j \leqslant J''}$ in $\mathbb R$ and integers $(k_j)_{1 \leqslant j \leqslant J}$, $(k'_{j})_{1 \leqslant j \leqslant J'}$, $(k''_{j})_{1 \leqslant j \leq J''}$ as in Theorem \ref{T:Poisson}.\\ 
Using notations in \cite{FK2011}, one picks $L$ and $l$ such that 
$(2L+1)=(2l+1)(2l'+1)$ where $cL^{\alpha}\leqslant l \leqslant L^{\alpha}/c$ for some $\alpha \in (0,1)$ and $c>0.$ Then, one decomposes 
\begin{equation*}\label{eq:5.0} 
\Lambda:=[-L,L] = \bigcup_{|\gamma|\leqslant l'} \Lambda_l(\gamma) 
\end{equation*}
where 
$\Lambda_l(\gamma):=(2l+1)\gamma +\Lambda_l.$\\ 
Next, for $\Lambda'\subset \Lambda,$ $U\subset \mathbb R$ and $E>0,$ one defines the following Bernoulli r.v.
\begin{equation}\label{eq:5.1}
X(E, U, \Lambda'):=
\begin{cases}
1 & \text{if $H_{\omega}(\Lambda')$ has at least one eigenvalue}\\ 
  &\text{in $E+(\nu(E)|\Lambda|)^{-1}U$,}\\
0 & \text{otherwise}
\end{cases}
\end{equation} 
and put $\Sigma(E, U, l):= 
\sum_{|\gamma|\leqslant l'} X(E, U, \Lambda_l(\gamma)).$
\vskip 1pt
\noindent First of all,  [Lemma 3.2, \cite{FK2011}] is the first ingredient of the proof which tell us that we can actually reduce our problem to consider eigenvalues of finite-volume operators restricted on much smaller intervals. This lemma is still true in the $n$-energy case for all $n\geq 2$.\\
Hence, to complete the proof of the stochastic independence w.r.t. three processes, one only need show that the quantity
\begin{equation}\label{eq:5.2}
\mathbb P \left (
\left \{
\omega; 
\begin{matrix}
&\Sigma(E, U_1, l)=k_1, &\ldots, \Sigma(E, U_j, l)=k_J \\ 
&\Sigma(E', U'_1, l)=k'_1, &\ldots, \Sigma(E', U'_{J'}, l)=k'_{J'}\\    
&\Sigma(E", U''_1, l)=k''_1, &\ldots, \Sigma(E'', U''_{J''}, l)=k''_{J''}
 \end{matrix}
\right\}
\right )
\end{equation}
should be approximated by the product 
\begin{align*}
&\mathbb P \left (
\left \{
\omega; 
\begin{matrix}
&\Sigma(E, U_1, l)=&k_1\\
 &\vdots &\vdots\\
&\Sigma(E, U_j, l)=&k_J 
 \end{matrix}
\right\}
\right ) \times 
\mathbb P \left (
\left \{
\omega; 
\begin{matrix}
&\Sigma(E', U'_1, l)=&k'_1\\
 &\vdots &\vdots\\
&\Sigma(E', U'_{J'}, l)=&k'_{J'}
 \end{matrix}
\right\}
\right )\\
&\times \mathbb P \left (
\left \{
\omega; 
\begin{matrix}
&\Sigma(E'', U''_1, l)=&k''_1\\
 &\vdots &\vdots\\
&\Sigma(E'', U''_{J''}, l)=&k''_{J''}
 \end{matrix}
\right\}
\right )\\
\end{align*}
as $L$ goes to infinity. Indeed, if the above statement is proved, Theorem \ref{T:oldpoisson} and Lemma 3.2 in \cite{FK2011} implies Theorem \ref{T:Poisson1}.\\
By a standard criterion of the convergence of point processes (c.f. e.g.  Theorem 11.1.VIII, \cite{DJD}), the above statement holds true if the following quantity vanishes for all real numbers $t_j, t'_{j'}, t''_{j''}$ when $L$ goes to infinity: 
\begin{align*}
&\mathbb E \left ( e^{-\sum_{j=1}^J t_j \Sigma(E, U_j, l) -\sum_{j'=1}^{J'} t'_{j'} \Sigma(E', U'_{j'}, l) -\sum_{j''=1}^{J''} t''_{j''} \Sigma(E'', U''_{j''}, l)       } \right) -\\
&\mathbb E \left ( e^{-\sum_{j=1}^J t_j \Sigma(E, U_j, l)}\right)
\mathbb E \left ( e^{-\sum_{j'=1}^{J'} t'_{j'} \Sigma(E', U'_{j'}, l) } \right)\mathbb E \left ( e^{-\sum_{j''=1}^{J''} t''_{j''} \Sigma(E'', U''_{j''}, l) } \right).
\end{align*}
\noindent From the fact that $\{\Lambda(\gamma)\}_{|\gamma| \leqslant l'}$ are pairwise disjoint intervals, operators $\{H_{\omega}(\Lambda(\gamma))\}_{|\gamma|\leqslant l'}$ are independent operator-valued r.v.'s. We thus have
\begin{align*}
&\mathbb E \left ( e^{-\sum_{j=1}^J t_j \Sigma(E, U_j, l) -\sum_{j'=1}^{J'} t'_{j'} \Sigma(E', U'_{j'}, l) -\sum_{j''=1}^{J''} t''_{j''} \Sigma(E'', U''_{j''}, l)       } \right)\\
&= \mathbb E \left| \prod_{|\gamma|\leqslant l'} e^{-\sum_j t_j X(E, U_j, \Lambda_l(\gamma)) -\sum_{j'} t'_{j'} X(E', U'_{j'}, \Lambda_l({\gamma})) -\sum_{j''} t''_{j''} X(E'', U''_{j''}, \Lambda_l(\gamma))} \right|\\
&= \prod_{|\gamma|\leqslant l'} \mathbb E  \left|  e^{-\sum_j t_j X(E, U_j, \Lambda_l(\gamma)) -\sum_{j'} t'_{j'} X(E', U'_{j'}, \Lambda_l({\gamma})) -\sum_{j''} t''_{j''} X(E'', U''_{j''}, \Lambda_l(\gamma))} \right|.
\end{align*}
 Our goal is to approximate terms of the form 
 $$\mathbb E  \left(  e^{-\sum_j^J t_j X(E, U_j, \Lambda_l(\gamma)) -\sum_{j'}^{J'} t'_{j'} X(E', U'_{j'}, \Lambda_l({\gamma})) -\sum_{j''}^{J''} t''_{j''} X(E'', U''_{j''}, \Lambda_l(\gamma))} \right)$$
 by the product 
 $$\prod_j^J \mathbb E e^{-t_j X(E, U_j, \Lambda_l(\gamma))}\prod_{j'}^{J'} \mathbb E e^{-t_j X(E', U'_{j'}, \Lambda_l({\gamma}))}\prod_{j''}^{J''} \mathbb E e^{-t''_{j''} X(E'', U''_{j''}, \Lambda_l(\gamma))}$$
as $L$ large enough with the remark that these r.v.'s are not independent.\\
To illustrate our computation for that, we will consider just three arbitrary Bernoulli r.v.'s $X_i=X(E_i, U_i, \Lambda'), \; i \in \{1, 2, 3\}$ and compute explicitly $\mathbb E(e^{\sum_{i=1}^3 a_i X_i})$. 
\begin{lemma}\label{L:compute}
 Let $\{E_i\}_{i=\overline{1,3}}$ be three positive, distinct energies in the localized regime and 
 $\{U_i\}_{i=\overline{1,3}}$ be three 
 non-empty compact intervals in $\mathbb R.$ 
 Assume that $\Lambda'$ is a sub-interval of $\Lambda=[-L, L]$ such that $|\Lambda'|=2l+1=O(L^{\alpha})$ for some $\alpha \in (0,1)$. Consider three Bernoulli r.v.'s $X_i=X(E_i, U_i, \Lambda')$ defined as in (\ref{eq:5.1}), we have 
 \begin{equation}\label{eq:5.3}
 \mathbb E(e^{\sum_{i=1}^3 a_i X_i})=\prod_{i=1}^3 \mathbb E(e^{a_i X_i})+ O\left((l/L)^{1+\theta}\right).
 \end{equation}
 with any $\theta \in (0,1)$ as $L$ large enough .
 \end{lemma}
 \begin{proof}[Proof of Lemma \ref{L:compute}]
Put $A:=\mathbb E(e^{\sum_{i=1}^3 a_i X_i}),$ we have
\begin{align*}
A&= \mathbb P(X_1=X_2=X_3=0) + \sum_{i=1}^3 e^{a_i} \mathbb P
\left(
\begin{matrix}
&X_i=1,\\ 
&X_j=0\;\forall j\ne i
\end{matrix}
\right)\\
&+\sum_{i<j} e^{a_i+a_j} 
\mathbb P\left(
\begin{matrix}
&X_i=X_j=1,\\
&X_k=0; k\ne i,j
\end{matrix}
\right) 
+e^{\sum_{i=1}^3 a_i}\mathbb P(\bigcap_{i=1}^3 \{X_i=1\}).
\end{align*}  
First of all, we rewrite 
\begin{align*} 
&\mathbb P(X_1=X_2=X_3=0)= 1-\sum_{i=1}^3 \mathbb P \left(
\begin{matrix}
&X_i=1,\\ 
&X_j=0\;\forall j\ne i
\end{matrix}
\right)\\
&-\sum_{i<j} \mathbb P\left(
\begin{matrix}
&X_i=X_j=1,\\
&X_k=0; k\ne i,j
\end{matrix}
\right) 
- \mathbb P(\bigcap_{i=1}^3 \{X_i=1\})
\end{align*}
and obtain that
\begin{align*}
&A=1+\sum_{i=1}^3 (e^{a_i}-1) \mathbb P \left(
\begin{matrix}
&X_i=1,\\ 
&X_j=0\;\forall j\ne i
\end{matrix}
\right)\\
&+\sum_{i<j} (e^{a_i+a_j}-1) \mathbb P\left(
\begin{matrix}
&X_i=X_j=1,\\
&X_k=0; k\ne i,j
\end{matrix}
\right)
+(e^{\sum_{i=1}^3 a_i}-1)\mathbb P(\bigcap_{i=1}^3 \{X_i=1\}).
\end{align*}
Next, use
\begin{align*}
&\mathbb P \left(
\begin{matrix}
&X_i=1,\\ 
&X_j=0\;\forall j\ne i
\end{matrix}
\right)
= \mathbb P(X_i=1)-\mathbb P\left(
\begin{matrix}
&X_i=X_{i+1}=1,\\
&X_{i+2}=0
\end{matrix}
\right)\\
 &-\mathbb P\left(
\begin{matrix}
&X_i=X_{i-1}=1,\\
&X_{i+1}=0
\end{matrix}
\right) 
-\mathbb P(\bigcap_{i=1}^3 \{X_i=1\})
\end{align*}
to get 
\begin{align*}
&A=1+\sum_{i=1}^3(e^{a_i}-1) \mathbb P(X_i=1) +\sum_{i<j} (e^{a_i}-1)(e^{a_j}-1)
\mathbb P\left(
\begin{matrix}
X_i=X_j=1,\\ 
X_k=0\; k\ne i,j
\end{matrix}
\right)\\
&+\left(e^{\sum_{i=1}^3 a_i} -1 -\sum_{i=1}^3(e^{a_i}-1)\right)\mathbb P(\bigcap_{i=1}^3 \{X_i=1\}). 
\end{align*}
Using a similar expansion for all terms of the form $\mathbb P\left(
\begin{matrix}
X_i=X_j=1,\\ 
X_k=0\; k\ne i,j
\end{matrix}
\right),\\
$
we obtain the following formula
\begin{align}\label{eq:5.7}
A&=1+\sum_{i=1}^3(e^{a_i}-1) \mathbb P(X_i=1) +\sum_{i<j} (e^{a_i}-1)(e^{a_j}-1)
\mathbb P\left(
X_i=X_j=1
\right) \notag\\
 &+\prod_{i=1}^3(e^{a_i}-1)\mathbb P (\bigcap_{i=1}^3 \{X_i=1\}).
\end{align}
On the other hand, from the observation that  
  $$\mathbb E e^{a_j X_j}=1+ (e^{a_j}-1)\mathbb P(X_j=1) \; \forall j=\{1,2,3\},$$
  we multiply the three equalities above to get 
\begin{align}\label{eq:5.8}
&\prod_{i=1}^3 \mathbb E(e^{a_i X_i})= 
1+\sum_{i=1}^3(e^{a_i}-1) \mathbb P(X_i=1) \notag \\
&+\sum_{i<j} (e^{a_i}-1)(e^{a_j}-1) \mathbb P(X_i=1)\mathbb P(X_j=1)
 +\prod_{i=1}^3(e^{a_i}-1)\mathbb P(X_i=1).
\end{align}
Hence, thanks to (\ref{eq:5.7}) and (\ref{eq:5.8}), we have 
\begin{align}\label{eq:5.4}
&\mathbb E\left(e^{\sum_{i=1}^3 a_i X_i}\right)-\prod_{i=1}^3 \mathbb E(e^{a_i X_i})=\notag\\
&\sum_{i<j} (e^{a_i}-1)(e^{a_j}-1)\left[ \mathbb P(X_i=X_j=1) -\mathbb P(X_i=1)\mathbb P(X_j=1)\right]\\
&+ \prod_{i=1}^3(e^{a_i}-1) \left[ \mathbb P(\bigcap_{j=1}^3 \{X_j=1\}) -\prod_{j=1}^3\mathbb P(X_j=1)\right] \notag.
\end{align}
Theorem \ref{T:decorrelation} and Theorem \ref{T:Wegner} yield, for any $\theta \in(0,1),$  
\begin{align}\label{eq:5.5}
&\mathbb P(X_i=X_j=1)+\mathbb P(X_i=1) \mathbb P(X_j=1) \notag\\
&\leqslant  C (l/L)^2(e^{(\log L)^{\beta}}+1)\leqslant C(l/L)^{1+\theta}
\end{align}
and
\begin{align}\label{eq:5.6}
&\mathbb P(\bigcap_{i=1}^3 \{X_i=1\})+\mathbb P(X_1=1) \mathbb P(X_2=1)\mathbb P(X_3=1) \notag\\
&\leqslant  C (l/L)^2(e^{(\log L)^{\beta}}+(l/L))\leqslant C(l/L)^{1+\theta}
\end{align}
 with $L$ large enough. Note that $C>0$ is a constant depending only on positive, fixed energies $\{E_i\}_{i=1}^3$.\\
From (\ref{eq:5.4})-(\ref{eq:5.6}), Lemma \ref{L:compute} follows.
\end{proof}
\noindent It is easy to see that the computation in Lemma \ref{L:compute} can apply to any finite number of Bernoulli r.v.'s defined as in (\ref{eq:5.1}). More precisely, for each $|\gamma| \leqslant l',$ we have 
\begin{align*}
&\mathbb E  \left(  e^{-\sum_j t_j X(E, U_j, \Lambda_l(\gamma)) -\sum_{j'} t'_{j'} X(E', U'_{j'}, \Lambda_l({\gamma})) -\sum_{j''} t''_{j''} X(E'', U''_{j''}, \Lambda_l(\gamma))} \right)\\
&=\prod_j^J \mathbb E e^{-t_j X(E, U_j, \Lambda_l(\gamma))}\prod_{j'}^{J'} \mathbb E e^{-t'_{j'} X(E', U'_{j'}, \Lambda_l({\gamma}))}\prod_{j''}^{J''} \mathbb E e^{-t''_{j''} X(E'', U''_{j''}, \Lambda_l(\gamma))}\\
& \times (1+O(l/L)^{1+\theta}).
\end{align*}
On the other hand, we also have similar formulas for $\mathbb Ee^{-\sum_j t_j X(E, U_j, \Lambda_l(\gamma))}$,\\ 
$\mathbb Ee^{-\sum_{j'} t'_{j'} X(E', U'_j, \Lambda_l(\gamma))}$ and $\mathbb E e^{-\sum_{j''} t''_j X(E'', U''_{j''}, \Lambda_l(\gamma))}.$\\
We therefore infer that:   
\begin{align*}
&\mathbb E  \left(  e^{-\sum_j t_j X(E, U_j, \Lambda_l(\gamma)) -\sum_{j'} t'_{j'} X(E', U'_{j'}, \Lambda_l({\gamma})) -\sum_{j''} t''_{j''} X(E'', U''_{j''}, \Lambda_l(\gamma))} \right)\\
&= \mathbb Ee^{-\sum_j t_j X(E, U_j, \Lambda_l(\gamma))} \mathbb E e^{-\sum_{j'} t'_{j'} X(E', U'_{j'}, \Lambda_l({\gamma}))} \mathbb E e^{-\sum_{j''} t''_{j''} X(E'', U''_{j''}, \Lambda_l(\gamma))}\\
& \times (1+O(l/L)^{1+\theta}).
\end{align*}
We have an observation that $|\gamma| \leqslant l'$ where $l'=O(L/l)$. 
Hence, by multiplying all above equalities side by side over $|\gamma|\leq l'$, we obtain that:
$$\mathbb E \left ( e^{-\sum_{j=1}^J t_j \Sigma(E, U_j, l) -\sum_{j'=1}^{J'} t'_{j'} \Sigma(E', U_{j'}, l) -\sum_{j''=1}^{J''} t''_{j''} \Sigma(E'', U_{j''}, l)       } \right)$$
is equal to the product of 
\begin{align*}
\mathbb E \left ( e^{-\sum_{j=1}^J t_j \Sigma(E, U_j, l)}\right)
\mathbb E \left ( e^{-\sum_{j'=1}^{J'} t'_{j'} \Sigma(E', U_{j'}, l) } \right)\mathbb E \left ( e^{-\sum_{j''=1}^{J''} t''_{j''} \Sigma(E'', U_{j''}, l) } \right)
\end{align*}
and an error term of the form $(1+x^{1+\theta})^{1/x}$ with $x=O(l/L)$.\\
Note that the above error term tends to 1 as $L$ goes to infinity. Hence, the stochastic independence for three point processes w.r.t. three distinct energies is proved. Finally, it is not hard to see that we can adapt this proof for $n-$energy case with any $n\geq 2$. 
\end{proof}
\appendix
\section{}\label{S:A}
\begin{proof}[Compute the determinant of matrix $A_0$ in (\ref{eq:A0})]
Put $A_0=(a_{ij}),$ we'll give here some details of computing the determinant of $A_0$ by hand (A mathematical software like Maple or Mathematica might be useful for checking the final result of this computation).\\   
First, expand this determinant by its sixth and last column and then by its first column to get 
  $$ |\det A_0| =\omega_{n-2}|\det B_0| $$
where $B_0$ is the $7\times 7$ matrix defined by
\begin{equation*}
\begin{pmatrix}
1&-1&0&0&1&-1&0\\ 
0&1&-1&0&0&1&-1\\ 
-\omega_{n-1}&\omega_{n-1}+\omega_{n}-E&-\omega_n&0&0&0&0\\ 
0&\dfrac{E}{E'}&0&0&0&-1&0\\ 
0&-\omega_n&\omega_n+\omega_{n+1}-E&-\omega_{n+1}&0&0&0\\ 
0&\omega_n&\omega_{n+1}-\omega_n&-\omega_{n+1}&0&0&E'\\
\dfrac{E}{E'}&0&0&0&-1&0&0 
\end{pmatrix}
\end{equation*}
 Now, we compute the determinant of $B_0.$\\
Take the sixth row minus the fifth row and take the first row plus the last row. Next, multiply the second row by $E'$ and take it plus the sixth row. Finally, expand the determinant of $B_0$ by the forth, the fifth and the last column to get   
$$ |\det A_0| =
\omega_{n-2}\omega_{n+1}|\det C_0| $$
where $C_0$ is the $4\times 4$ matrix defined by
\begin{equation*}
\begin{pmatrix}
1+\dfrac{E}{E'}&-1&0&-1\\ 
0&E'+2\omega_n&E-E'-2\omega_n&E'\\ 
-\omega_{n-1}&\omega_{n-1}+\omega_{n}-E&-\omega_n&0\\ 
0&\dfrac{E}{E'}&0&-1\\
\end{pmatrix}   
\end{equation*}
Finally, by an explicit computation for the determinant of $C_0$, we obtain that 
$$ |\det A_0| = \dfrac{4E}{E'}(E+E')\omega_{n-2}\omega_{n+1} \Big|\omega_{n} +\dfrac{E'-E}{4}\Big|.$$ 
\end{proof}

\begin{proof}[Compute the determinant of matrix $A_1$ in (\ref{eq:A1})]
The determinant of $A_1$\\ 
can be computed as follows:
First, expand this determinant by its sixth and last column and then by its fifth and first column to get 
  $$ |\det A_1| =\omega_{n-2}\omega_{n+1}|\det B_1| $$
where $B_1$ is the $6\times 6$ matrix defined by
\begin{equation*}
\begin{pmatrix}
1&-1&0&-1&1&0\\ 
0&1&-1&0&1&-1\\ 
-\omega_{n-1}&\omega_{n-1}+\omega_{n}-E&-\omega_n&0&0&0\\ 
\omega_{n-1}&\omega_{n}-\omega_{n-1}&-\omega_n&0&-E'&0\\ 
0&0&\dfrac{E}{E'}&0&0&-1\\ 
-\dfrac{E}{E'}&0&0&-1&0&0\\ 
\end{pmatrix}
\end{equation*}
Second, take the first row of matrix B minus its last row and take the second row minus the fifth row, then expand the determinant of $B_1$ by its forth and sixth column to obtain
$$ |\det A_1| =  \omega_{n-2}\omega_{n+1}|\det C_1|$$
where $C_1$ is the following $4\times 4$ matrix
\begin{equation*}
\begin{pmatrix}
1+\dfrac{E}{E'}&-1&0&1\\ 
0&1&-1-\dfrac{E}{E'}&1\\ 
-\omega_{n-1}&\omega_{n-1}+\omega_{n}-E&-\omega_n&0\\ 
\omega_{n-1}&\omega_n-\omega_{n-1}&-\omega_n&-E'\\ 
\end{pmatrix}   
\end{equation*}
Finally, by an explicit computation for the determinant of the matrix $C_1$, we infer that
$$ |\det A_1| = \dfrac{4E}{E'}\omega_{n-2}\omega_{n+1} \Big|\omega_{n-1}\omega_n -\dfrac{(E+E')^2}{4}\Big|.$$ 
\end{proof}

{\small
\noindent \emph{Acknowledgements.} 
This work would have been impossible without great help of my advisor, Prof. Fr\'{e}d\'{e}ric Klopp. The author gratefully acknowledges the collaboration and various discussions with him. I am grateful for some corrections and helpful suggestions made by an anonymous referee which definitely improved the presentation of this work.
}

\noindent {\tiny (Trinh Tuan Phong) Laboratoire Analyse, G\'{e}om\'{e}trie et Applications,\\ 
	UMR 7539, Institut Galil\'{e}e, Universit\'{e} Paris 13,\\ 
	99 avenue J.B. Cl\'{e}ment, 93430 Villetaneuse, France\\
\emph{Email}: trinh@math.univ-paris13.fr 
}
\end{document}